\documentclass{SciPost}

\hypersetup{
    colorlinks,
    linkcolor={red!50!black},
    citecolor={blue!50!black},
    urlcolor={blue!80!black}
}

\usepackage{bbm}
\usepackage{mathtools,bm}
\usepackage{graphicx}
\usepackage{booktabs}
\usepackage{microtype}
\usepackage[normalem]{ulem}
\usepackage[bitstream-charter]{mathdesign}
\DeclareSymbolFont{usualmathcal}{OMS}{cmsy}{m}{n}
\DeclareSymbolFontAlphabet{\mathcal}{usualmathcal}

\newcommand{\dd}{\mathrm d}
\newcommand{\RP}{\mathbb{RP}}
\newcommand{\cA}{\mathcal A}

\newcommand{\cK}{\mathcal K}
\newcommand{\cO}{\mathcal O}
\newcommand{\abs}[1]{\left|#1\right|}

\newcommand{\Area}{\operatorname{Area}}
\newcommand{\ket}[1]{\left|#1\right\rangle}
\newcommand{\bra}[1]{\left\langle#1\right|}

\usepackage{amsthm}
\newtheorem{thm}{Theorem}[section]
\newtheorem{prop}[thm]{Proposition}

\newcommand{\Sbb}{\mathbb{S}}
\newcommand{\Kbb}{\mathbb{K}}
\newcommand{\Rbb}{\mathbb{R}}
\newcommand{\M}{\mathcal{M}}
\renewcommand{\S}{\mathcal{S}}

\renewcommand{\v}[1]{\ensuremath{\mathbf{#1}}}

\fancypagestyle{SPstyle}{
\fancyhf{}

\fancyfoot[C]{\textbf{\thepage}}
}

\begin{document}

\pagestyle{SPstyle}

\begin{center}{\Large \textbf{\color{scipostdeepblue}{
A Geodesic Route toward Holography beyond AdS: Tessellating the Schwarzschild Black Hole
}}}\end{center}

\begin{center}\textbf{
Shi Gu\textsuperscript{2$\star$}, Qiang Wen\textsuperscript{1$\dagger$}, Haocheng Zhong\textsuperscript{1$\ddagger$} and Yiwei Zhong\textsuperscript{1$\ddagger$}\footnote{The authors are listed in alphabetical order. The ordering of the affiliations is unconventional, reflecting the special graduate degree requirements of Southeast University.}
}\end{center}

\begin{center}
{\bf 1} School of Physics and Shing-Tung Yau Center, Southeast University, Nanjing, China
\\
{\bf 2} College of Computer Science and Technology, Zhejiang University, Hangzhou, China
\\
$\star$ \href{mailto:gus@zju.edu.cn}{\small gus@zju.edu.cn},  
$\dagger$ \href{mailto:wenqiang@seu.edu.cn}{\small wenqiang@seu.edu.cn}
\\
$\ddagger$ These two authors contributed equally to this work.
\end{center}

\section*{\color{scipostdeepblue}{Abstract}}
\textbf{We show that geodesic networks can perfectly tessellate a broad class of Riemannian geometries far beyond homogeneous manifolds, providing a representation of spatial geometry in which the areas of arbitrary smooth hypersurfaces are reproduced by weighted geodesic crossings. The Schwarzschild exterior and the Einstein--Rosen bridge furnish explicit black-hole realizations, with complete reconstruction of the bridge requiring both asymptotic boundaries. Building on this geometric structure, we construct partial entanglement entropy (PEE) tensor-network models beyond anti-de~Sitter (AdS) space. For Schwarzschild, these models realize the Bekenstein--Hawking area law and an Einstein--Rosen (ER) = Einstein--Podolsky--Rosen (EPR) picture, with horizon entropy and entanglement between the two asymptotic systems governed by the same throat area. They also provide a geometric foundation for a Ryu--Takayanagi-type entropy prescription. In this framework, geometric area and network entanglement capacity are expressed by the same crossing measure, making their connection manifest. Our results offer a geodesic-network perspective on gravitational geometry and a foundation for holographic models beyond AdS.
}

\vspace{10pt}
\noindent\rule{\textwidth}{1pt}
{\setlength{\cftbeforesecskip}{4pt}\tableofcontents}
\noindent\rule{\textwidth}{1pt}
\vspace{10pt}

\section{Introduction}
\label{sec:intro}

The anti-de~Sitter (AdS)/conformal field theory (CFT) correspondence \cite{Maldacena:1997re,Gubser:1998bc,Witten:1998qj}
is the most concrete realization of the holographic principle: a
$(d+1)$-dimensional quantum theory of gravity in anti-de~Sitter space is
dual to a $d$-dimensional conformal field theory without gravity, residing
on the boundary of AdS.  This duality provides a non-perturbative definition
of quantum gravity in asymptotically AdS spacetimes. In the context of AdS/CFT, the Ryu--Takayanagi (RT) formula \cite{Ryu:2006bv,Ryu:2006ef} relates the entanglement entropy of a boundary subregion
$A$ to the area of a minimal surface $\Sigma_A$ homologous to $A$ in the
bulk:
\begin{equation}
  S(A) = \frac{\Area(\Sigma_A)}{4G}.
  \label{eq:RT}
\end{equation}
The RT formula and its generalizations \cite{Hubeny:2007xt,Lewkowycz:2013nqa,Faulkner:2013ica,Engelhardt:2014gca} make boundary entanglement a probe of bulk geometry and motivate the idea that geometry emerges from quantum correlations.

The connection between geometry and entanglement was made concrete through
two complementary frameworks. The \textit{Crofton reconstruction} framework \cite{Santalo,Helgason} was introduced into AdS/CFT by Czech et al. \cite{CzechEtAl}, who reformulated the RT formula in the language of integral geometry (see also \cite{Czech:2015tensor,Czech:2016stereo,Czech:2017zfq,Czech:2019time} for developments along this line). Concretely, consider a $d$-dimensional Riemannian manifold $\mathcal{M}$ that
admits a kinematic space, and let $\Sigma$ be an arbitrary $(d-1)$-dimensional
hypersurface in $\mathcal{M}$.  The Crofton formula
states that the area of $\Sigma$ can be evaluated by counting the intersections
between $\Sigma$ and the geodesics in the kinematic space:
\begin{equation}
\Area(\Sigma)=\frac{d-1}{\Omega_{d-2}}\int_{\mathbb{K}_{\mathcal{M}}}\#(\gamma\cap\Sigma)\,\dd\nu_{\mathcal M},
\label{eq:crofton-general}
\end{equation}
where $\gamma$ denotes an unoriented geodesic in the kinematic space $\mathbb{K}_{\mathcal{M}}$, $\dd\nu_{\mathcal M}$ is the
kinematic measure, $\Omega_{d-2}$ is the area of the unit $(d-2)$-sphere, and $\#(\gamma\cap\Sigma)$ counts the number of times
$\gamma$ pierces $\Sigma$.  Equivalently, the intersection density between
$\Sigma$ and the geodesics is constant on $\Sigma$. On a time-symmetric spatial slice of AdS$_{d+1}$, this measure is inherited from the $SO(1,d)$ isometry group of hyperbolic space $\mathbb H^d$, and in the context of AdS/CFT, it can be identified with the so-called \textit{partial entanglement entropy} (PEE) structure of the boundary dual CFT \cite{Lin:2024dho}.\footnote{In AdS$_3$, the PEE structure of the dual CFT$_2$ reduces to the conditional mutual information, which coincides with the observation in \cite{CzechEtAl}.}  

The \textit{tensor network} framework \cite{Swingle:2009bg,HaPPY,HaydenEtAl}
provides a discrete, constructive toy model for the AdS/CFT correspondence. In these
models, the bulk geometry is discretized by a graph whose edges carry entangled qudits. For suitable tensors, boundary entropies are controlled by minimal cuts, giving a network analogue of the RT formula. In the appropriate large-bond-dimension limit, the random tensor networks \cite{HaydenEtAl} in particular capture key features of the RT formula, including the phase transitions of the RT surfaces and the entanglement wedge reconstruction.

A unified framework emerged from the PEE thread program \cite{Wen:2018whg,Wen:2020ech,Lin:2023rxc,Han:2019scu}.  The
PEE $\mathcal{I}(A,B)$ is an additive measure of two-body correlations with properties analogous to mutual information \cite{Casini:2008wt,Wen:2019iyq}. See \cite{Lin:2024dho} for a more accurate description of the physical requirements that determine the PEE for the vacuum state of a CFT.  In Ref.~\cite{Lin:2023rxc} it was shown that any two-point PEE $\mathcal{I}(\mathbf{x},\mathbf{y})$ can be geometrized by
a bulk geodesic connecting the boundary points $\mathbf{x}$ and $\mathbf{y}$, while the value of $\mathcal{I}(\mathbf{x},\mathbf{y})$ determines the density distribution of the geodesics. Such geodesics are called the PEE threads.  The set of all the PEE threads forms a continuous network that perfectly tessellates the AdS geometry, with the thread flux through any area element in the AdS bulk being a constant $1/4G$. In this language, the RT surface for a region $A$ minimizes the weighted intersection count among bulk surfaces homologous to $A$.

This construction was extended in several directions.
Refs.~\cite{Lin:2023rxc,Basu:2026btz,Basu:2026sub} established the connection to bit threads \cite{FreedmanHeadrick}, showing
that, for any static interval or spherical region, a unique bit thread configuration emerges from the PEE thread configuration.  Ref.~\cite{Lin:2024dho} further proved that the PEE threads coincide with the geodesics in the kinematic space, and reproduced the Crofton formula in AdS based on the PEE structure of the boundary CFT.  Ref.~\cite{WenXuZhong} assigned quantum states to the vertices of the PEE network and constructed factorized, Pastawski--Yoshida--Harlow--Preskill (HaPPY)-like, and random PEE tensor networks. The factorized model reproduces the RT entropy for spherical regions, while the random construction extends the entropy calculation to generic regions in the appropriate large-bond-dimension limit. Ref.~\cite{BasuWen} analyzed the kinematic space for \emph{gravitational subregions} in AdS, showing that PEE threads emanating from a codimension-one surface perfectly cover the subregion, and built holographic tensor network models on the confined PEE thread network, realizing features of the surface-state correspondence and the generalized entanglement wedges within the stated geometric and tensor-network assumptions.

Holographic applications of Crofton reconstruction have mainly focused on homogeneous spatial geometries, such as the hyperbolic slices of vacuum AdS, where symmetry provides a convenient kinematic measure. The additional question for a boundary-generated network is whether the chosen source reaches all geodesic directions in the target region.  Here we consider spherically symmetric, asymptotically flat spaces, including the Schwarzschild black hole. Their $SO(3)$ isometry group alone does not determine the required distribution over radial positions and geodesic directions. We instead obtain the chord measure from the intrinsic boundary flux and geodesic transport, and show that an outer sphere supplies a complete Crofton reconstruction of the exterior.

Our strategy is to start from an outer sphere $\S_{+}$ of radius $R_+$ and 
the intersections $(x,[\mathbf v])$ defined on it. Here $x$ labels a point of
$\S_{+}$, and $[\mathbf v]$ is the unoriented tangent direction of a geodesic that intersects
$\S_{+}$ at $x$. We assign to these intersections the intrinsic
measure \cite{BasuWen},
\begin{equation}\label{eq:intrinsic}
	\dd\Gamma_+
	=\abs{\mathbf v\cdot\mathbf n_+}\,\dd A_+\,
	\dd\Omega_{[\mathbf v]}.
\end{equation}
Here $\mathbf n_+$ denotes the inward unit normal vector to $\S_{+}$ and $\dd\Omega_{[\mathbf v]}$ is the rotation-invariant measure on the
projective plane $\RP^2$ of unoriented line directions $(\mathbf v\sim-\mathbf v)$, normalized to have total measure $2\pi$. Using a hemisphere representative, it is given by
$\dd\Omega_{[\mathbf v]}=\sin\theta\,\dd\theta\,\dd\phi$ with $\theta\in[0,\pi/2]$ and $\phi\in[0,2\pi)$.
With $\theta$ measured from the inward normal, the factor $\abs{\mathbf v\cdot\mathbf n_+}=\cos\theta$ converts the isotropic direction measure into crossing flux. Integrating over all directions gives
\begin{equation}\label{eq:intrinsicp}
\int\dd\Gamma_+=\Area(\S_{+})\int_{\RP^2}
\abs{\mathbf v\cdot\mathbf n_+}\,
\dd\Omega_{[\mathbf v]}=\pi\,\Area(\S_{+}),
\end{equation}
where we used $\int_{\RP^2}\abs{\mathbf v\cdot\mathbf n}\,\dd\Omega_{[\mathbf v]}=\pi$. 

This boundary normalization does not by itself guarantee the Crofton identity for interior surfaces. In this paper, we therefore start from the geodesic gas\footnote{By a ``geodesic gas'' we mean an uncountably infinite family of geodesics weighted by a measure on the space of geodesics, so that it carries a well-defined density and flux: counting the geodesics that cross a surface is equivalent to computing the flux of the gas through that surface.} launched inward from $\S_{+}$ with the intrinsic source measure~\eqref{eq:intrinsic}, and explicitly compute its intersections with other surfaces inside $\S_{+}$, verifying whether the Crofton formula still holds for them. The construction is summarized by: 
\[
\begin{gathered}
\boxed{\text{boundary intrinsic measure}}\\
\downarrow\\
\boxed{\text{boundary-sourced geodesic gas}}\\
\downarrow\\
\boxed{\text{intersection counting for surfaces}}\\
\downarrow\\
\boxed{\text{Crofton formula}}
\end{gathered}
\]
 The construction also extends to the Einstein--Rosen spatial bridge. A single boundary source misses a directional sector on the opposite side, whereas sources on both boundaries tessellate the full bridge between them when their shared crossing sector is counted only once.

The same surface-generated construction extends to perfect tessellation of subregions in more general Riemannian manifolds. On a smooth Riemannian manifold, one starts from a smooth, two-sided source surface $\Sigma$, equips its transverse geodesic crossings with the intrinsic flux measure, and transports this measure along the geodesic flow. Liouville invariance preserves the local crossing measure independently of homogeneity. To obtain a complete network in a target subregion $\mathcal W_{\Sigma}$, the sourced geodesics must represent every unoriented geodesic direction there, up to sets of measure zero. With this coverage and correct normalization, the chord measure reproduces the area of every smooth hypersurface in $\mathcal W_{\Sigma}$.  The local transport mechanism belongs to classical integral geometry \cite{Santalo}. The broader viewpoint is to describe spatial geometry by a measured network of geodesics.

To connect this geometric description to holography, we next ask how quantum states assigned to the network can encode boundary entanglement. AdS/CFT remains the best-understood realization of holographic duality \cite{Maldacena:1997re,Gubser:1998bc,Witten:1998qj}, but a comparably complete framework for general backgrounds is still lacking.  Some progress in this direction has been made in flat-space holography, both in its celestial and Carrollian formulations \cite{Bagchi:2010zz,LiTakayanagiFlat,Wen:2018mev,Jiang:2017ecm,Pasterski:2017celestial,Raclariu:2021celestial,Pasterski:2021lectures,Mason:2023carrollian,Bagchi:2025vri,Nguyen:2025zhg,Ruzziconi:2026bix,Zhu:2026celestial,Donnay:2023celestial,Pasterski:2021snowmass}. Other settings include Lifshitz holography \cite{Kachru:2008yh,Taylor2015} and warped AdS holography \cite{Anninos:2008fx,Detournay:2012pc,Song:2016gtd,Song:2016warped}, as well as the Kerr/CFT correspondence \cite{Guica:2008mu,Compere:2012jk} and holography with a $T\bar T$ deformation \cite{McGough:2016TTbar,Guica:2019tbar,Apolo:2019tsT,Apolo:2023glueon}. Here we extend the PEE tensor-network constructions of Refs.~\cite{WenXuZhong,BasuWen} to the surface-generated reconstruction regions described above.

In this framework, the area--entropy connection becomes manifest at the level of network cuts. After quantum degrees of freedom and bond entropy weights are assigned, the crossings that measure area also measure the entanglement capacity of a cut; appropriate states realize that capacity as entropy in the regimes described below.  The area measure alone does not fix a boundary quantum state or derive the gravitational entropy normalization, and generic RT entropy claims require the additional tensor-network and continuum-limit conditions discussed below. The present formalism concerns spatial geometry and its entanglement models; extending it to gravitational dynamics remains a further task.

The paper is organized as follows. Section~\ref{sec:2} constructs the boundary-sourced geodesic network for the Schwarzschild exterior and establishes Crofton reconstruction for arbitrary smooth surfaces in the exterior shell. Section~\ref{sec:bridge} then extends the construction to the Einstein--Rosen spatial bridge, identifying the missing directional sector of a one-sided source and showing how bilateral sources restore complete coverage. Section~\ref{sec:liouville transport} develops the general Liouville-transport formulation on smooth Riemannian manifolds, distinguishing local measure preservation from the global requirements on a chosen source. With this geometric framework in place, Section~\ref{sec:pee-shell} assigns quantum states to the networks and studies their entanglement properties, including the conditions needed for an RT interpretation beyond AdS. Finally, Section~\ref{sec:conclusion} discusses the implications and limitations of the construction and directions for further work. The proof of the transport theorem is given in Appendix~\ref{sec:app-transport}.

\section{Crofton reconstruction for the exterior of the Schwarzschild black hole}\label{sec:2}
\subsection{Geodesics sourced from a boundary sphere}
\label{sec:boundary}

Consider a spherically symmetric, time-symmetric spatial slice with a
general radial metric
\begin{equation}
  \dd s^2=\Lambda(r)^2\dd r^2+r^2\dd\Omega_2^2,
  \qquad \dd\Omega_2^2=\dd \theta^2+\sin^2 \theta \dd\phi^2,
  \label{eq:metric}
\end{equation}
where $\Lambda(r)$ is smooth, positive, and finite for $r>r_h$, with $\Lambda(r)\to\infty$ as $r\downarrow r_h$ at the horizon. The areal radius $r$ is monotonic throughout the exterior. We work within the annular shell
\begin{equation}
 \cA=\{R_-\le r\le R_+\},
 \qquad r_h<R_-<R_+,
 \label{eq:shell}
\end{equation}
with inner sphere $\S_{-}$ and outer sphere $\S_{+}$. The Schwarzschild geometry,
\begin{equation}
 \Lambda(r)^2=\Bigl(1-\frac{2Gm}{r}\Bigr)^{-1},\qquad r_h=2Gm,
 \label{eq:schwarzschild}
\end{equation}
is the primary case of interest, where we also call $\cA$ the Schwarzschild shell, but the constructions of
Sections~\ref{sec:boundary} and~\ref{sec:transport} apply to \emph{any} $\Lambda(r)$ with
monotonic $r$.

By spherical symmetry, every spatial geodesic lies in a plane through the center. With arclength $s$ and polar angle $\phi$ in that plane, the angular momentum
\begin{equation}
 J=\abs{r^2\dot\phi}
 \label{eq:J-def}
\end{equation}
is a constant of motion. Here the dot denotes $\dd/\dd s$. Since $s$ is arclength, the geodesic has unit speed in the metric~\eqref{eq:metric}:
\begin{equation}
 \Lambda(r)^2\dot r^{\,2}+\frac{J^2}{r^2}=1,
 \label{eq:unit-speed}
\end{equation}
so
\begin{equation}
 \dot r^{\,2}=\Lambda(r)^{-2}\Bigl(1-\frac{J^2}{r^2}\Bigr).
 \label{eq:rdot}
\end{equation}

For $J>r_h$, equation \eqref{eq:rdot} admits a turning point at $r=J$. Moreover, differentiating Eq.~\eqref{eq:rdot} away from the turning point and taking the smooth limit $r\to J$ gives
\begin{equation}
	\left.\ddot r\right|_{r=J}
	=\frac{1}{J\Lambda(J)^2}>0,
	\label{eq:turning-point-minimum}
\end{equation}
which confirms that $r=J$ is a radial minimum. In the Schwarzschild case, this reduces to $\left.\ddot r\right|_{r=J}=\bigl(1-r_h/J\bigr)/J$. If $0\leq J\leq r_h$, the zero $r=J$ lies at or inside the horizon, so there is no radial turning point in the exterior region $r>r_h$.

A shell chord is then defined to be a maximal connected segment of positive arclength of a unit-speed spatial geodesic $\gamma$ contained in $\mathcal A$. If the intersection $\gamma\cap\cA$ has more than one connected component, each component is regarded as a different chord. We denote the space of unoriented shell chords by $\mathbb K_{\mathcal A}$; its measure will be specified below.

Equation~\eqref{eq:rdot} gives the following chord classification, excluding the measure-zero tangent limit $J=R_+$ \footnote{The set with $J=R_+$ is tangent to the outer sphere and its intersection with $\mathcal A$ has zero arclength and hence is not a shell chord.}:
\begin{prop}\label{prop:shell chord classification}
	A spatial geodesic $\gamma$ can be classified according to the value of $J$:
	\begin{enumerate}
		\item $J>R_+$: the geodesic cannot enter the shell, so $\gamma\cap\cA=\emptyset$;
		
		\item $R_-\leq J<R_+$: the radial minimum $r=J$ lies in $\cA$, so $\gamma\cap\cA$ has both endpoints on $\S_{+}$;
		
		\item $0\leq J<R_-$: there is no turning point in $\cA$, and each connected component of $\gamma\cap\cA$ is monotonic in $r$ and joins $\S_{+}$ to $\S_{-}$.
	\end{enumerate}
	Therefore, the shell chords in $\mathbb K_{\mathcal A}$ should satisfy $0\leq J<R_+$.
\end{prop}
As illustrated in Fig.~\ref{fig:chords}, every shell chord has at least one endpoint on $\S_{+}$. We call the second type an outer--outer chord (or returning chord) and denote the space of all such chords by $\Kbb_{++}$, while $\Kbb_{-+}$ is the space of inner--outer chords in the third type.    The set of the shell chords is therefore
\begin{equation}
	\Kbb_{\cA}=\Kbb_{++}\cup \Kbb_{-+}.
\end{equation}

\begin{figure}[h]
	\centering
	\includegraphics[scale=0.8]{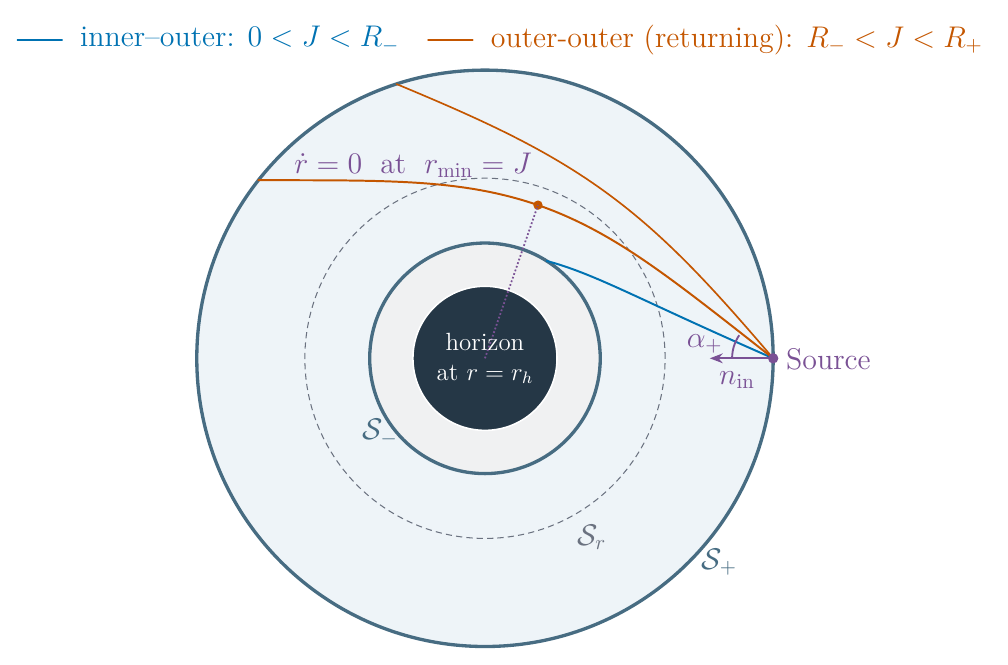}
	\caption{Inner--outer chords are shown in blue and outer--outer chords in orange. For an outer--outer chord, the intersection count $N_r$ with an intermediate sphere $\mathcal S_r$ depends on the turning radius.}	
	\label{fig:chords}	
\end{figure}

The shell is also nontrapping: every chord reaches its boundary endpoints in finite arclength. For an inner-outer chord with $0\leq J<R_-$, its length is given by \eqref{eq:rdot}:
\begin{equation}
 L_{-+}(J)=\int_{R_-}^{R_+}
 \frac{\Lambda(r)\,\dd r}{\sqrt{1-J^2/r^2}}.
 \label{eq:inner-outer-length}
\end{equation}
Since $\Lambda(r)$ is finite on the closed interval $[R_-,R_+]$ and $1-J^2/r^2$ is strictly positive there, the integrand is bounded and $L_{-+}(J)$ is finite. For an outer--outer chord with $R_-\leq J<R_+$, the total length is
\begin{equation}
 L_{++}(J)=2\int_J^{R_+}
 \frac{\Lambda(r)\,\dd r}{\sqrt{1-J^2/r^2}}.
 \label{eq:outer-outer-length}
\end{equation}
Near the turning point,
\begin{equation}
 \frac{1}{\sqrt{1-J^2/r^2}}
 =\sqrt{\frac{J}{2(r-J)}}
 \left[1+\mathcal O\!\left(\frac{r-J}{J}\right)\right].
 \label{eq:turning-point-integrability}
\end{equation}
The singularity is integrable and $\Lambda(J)$ is finite, so $L_{++}(J)$ is also finite.

The geodesic initial-value theorem assigns a unique oriented geodesic to each interior point and unit tangent direction. Its maximal connected segment in $\mathcal A$ is a shell chord. Thus every interior unoriented direction is represented in $\mathbb K_{\mathcal A}$; uniformity will depend on the chord measure.

We now parametrize the chords by inward data on $\S_{+}$. For the metric~\eqref{eq:metric}, the inward unit normal is
\begin{equation}
	\mathbf n_+=-\Lambda(R_+)^{-1}\partial_r.
	\label{eq:outer-inward-normal}
\end{equation}
The inward source section on $\S_{+}$ is
\begin{equation}
	C_+^{\mathrm{in}}
	\equiv \left\{(x,\mathbf v):x\in \S_{+},\;
	h(\mathbf v,\mathbf v)=1,\;
	h(\mathbf v,\mathbf n_+)>0\right\},
	\label{eq:source-section}
\end{equation}
where $h$ is the spatial metric and $\mathbf v$ is the inward unit tangent representative of the corresponding unoriented line direction $[\mathbf v]$. The source data in $C_+^{\mathrm{in}}$ label \emph{intersections} with $\S_{+}$, rather than distinct geodesic chords in $\mathcal A$. Let
\begin{equation}\label{eq:map}
	q:C_+^{\mathrm{in}}\longrightarrow\mathbb K_{\mathcal A}
\end{equation}
map a set of inward boundary data $z=(x,\mathbf v)$ to its corresponding maximal unoriented shell chord $\gamma=q(z)$. Completeness of this boundary parametrization is summarized by the following proposition:
\begin{prop}[Completeness of the outer source]
 \label{prop:outer-source-completeness}
 The chord projection $q:C_+^{\mathrm{in}}\to\mathbb K_{\mathcal A}$ is surjective. Equivalently,
 \begin{equation}
  q\bigl(C_+^{\mathrm{in}}\bigr)=\mathbb K_{\mathcal A}.
  \label{eq:outer-source-completeness}
 \end{equation}
\end{prop}

Proposition~\ref{prop:outer-source-completeness} allows us to define a measure on $\mathbb K_{\mathcal A}$ from the source measure on $C_+^{\mathrm{in}}$. We first define the source multiplicity by
\begin{equation}
	N_+(\gamma)\equiv \#q^{-1}(\gamma).
	\label{eq:source-multiplicity}
\end{equation}
Geometrically, $N_+(\gamma)$ counts the number of endpoints of $\gamma$ on $\S_{+}$. We equip $C_+^{\mathrm{in}}$ with the intrinsic measure \eqref{eq:intrinsic}, and the measure $\dd\nu_{\mathcal A}$ on $\mathbb K_{\mathcal A}$ is determined by the following equation:
\begin{equation}
	\int_{\mathbb K_{\mathcal A}}F(\gamma)\,
	\dd\nu_{\mathcal A}(\gamma)
	=\int_{C_+^{\mathrm{in}}}\frac{F(q(z))}{N_+(q(z))}\,
	\dd\Gamma_+(z),
	\label{eq:quotient}
\end{equation}
i.e., an inner-outer chord (one source intersection) keeps its weight,
while the weight of an outer--outer chord (two) is halved.  Equivalently,
\begin{equation}
	q_*\dd\Gamma_+=N_+(\gamma)\,
	\dd\nu_{\mathcal A}(\gamma),
	\label{eq:push-multiplicity}
\end{equation}
where $q_*$ denotes the push-forward of measures.\footnote{For a
	measurable map $f:X\to Y$ and a measure $\mu$ on $X$, the push-forward
	$f_*\mu$ is the image measure on $Y$ defined by
	$(f_*\mu)(E)=\mu(f^{-1}(E))$ for every measurable $E\subset Y$.}

 On $\S_{+}$, let $\alpha_+$ be the angle with the inward normal and $\psi_+$ the azimuth around the radial axis. The intrinsic measure~\eqref{eq:intrinsic} becomes
\begin{equation}
 \dd\Gamma_+
 =\cos\alpha_+\,\dd A_+\,
 \sin\alpha_+\dd\alpha_+\dd\psi_+.
 \label{eq:source-alpha}
\end{equation}
Equation~\eqref{eq:outer-inward-normal} and the radial component of the inward unit tangent $\mathbf v=\dot r\,\partial_r+\dots$ give
\begin{equation}
	\cos\alpha_+
	=\abs{\mathbf v|_{r=R_+}\cdot\mathbf n_+}
	=\sqrt{1-J^2/R_+^2},
	\quad \Rightarrow\quad J=R_+\sin\alpha_+,
\end{equation}
Substituting $\dd A_+=R_+^2\dd\Omega_x$ and $J\dd J=R_+^2\sin\alpha_+\cos\alpha_+\dd\alpha_+$ into Eq.~\eqref{eq:source-alpha} gives
\begin{equation}
 \dd\Gamma_+=J\,\dd J\,\dd\psi_+\,\dd\Omega_x,
 \label{eq:source-J}
\end{equation}
where $\dd\Omega_x$ is the standard solid-angle measure on the unit position two-sphere, normalized by $\int_{\S^2}\dd\Omega_x=4\pi$. Integrating over the launch point and azimuth gives
\begin{equation}
 \dd\Gamma_+^{(J)}=8\pi^2J\,\dd J.
 \label{eq:source-total}
\end{equation}
Integrating over $0<J<R_+$ gives $4\pi^2R_+^2=\pi\,\Area(\S_{+})$, in agreement with Eq.~\eqref{eq:intrinsicp}. We next test whether this normalization is preserved on interior surfaces.

\subsection{Reconstruction for spherical surfaces}
\label{sec:transport}

Consider an intermediate sphere $\mathcal S_r$ with $R_-<r<R_+$. For each chord, let $N_+$ and $N_r$ count its intersections with $\S_+$ and $\mathcal S_r$, respectively. The multiplicities are\footnote{The shell lies strictly outside the horizon, where $\Lambda(r)$ is finite and positive and the proper radial coordinate $\dd\rho=\Lambda(r)\dd r$ is monotonic in $r$. This one-sided description must be reconsidered when extending through the Einstein--Rosen throat in Section~\ref{sec:bridge}.}
\begin{center}
\begin{tabular}{c c c c}
\toprule
Range of $J$ & Type & $N_+$ & $N_r$\\
\midrule
$0\le J<R_-$ & inner-outer & $1$ & $1$\\
$R_-<J<r$ & outer--outer & $2$ & $2$\\
$r<J<R_+$ & outer--outer & $2$ & $0$\\
\bottomrule
\end{tabular}
\end{center}
We now consider the intersections of $\mathcal S_r$ with the geodesic family sourced from $\S_+$. These crossings inherit an induced measure from the fixed boundary source, with the source multiplicity removed as in Eq.~\eqref{eq:quotient}. We denote its $J$-marginal, integrated over crossing positions and azimuthal directions, by $\dd\Gamma_r^{\mathrm{ind},(J)}$. Counting the $N_r(J)$ crossings of each chord and using Eq.~\eqref{eq:source-total} gives
\begin{equation}
 \dd\Gamma_r^{\mathrm{ind},(J)}
 =\frac{N_r(J)}{N_+(J)}\,\dd\Gamma_+^{(J)}
 =8\pi^2J\,\dd J,
 \qquad 0<J<r.
 \label{eq:induced-J}
\end{equation}
The cancellation $N_r/N_+=1$ occurs separately in both sectors with non-zero $N_r$.

The rotational invariance of the whole configuration and the symmetries of the measure $\dd\Gamma_+$ imply that, for fixed $J$, the measure over the intersection point and the azimuth is proportional to $\dd\Omega_x\dd\psi$. Writing it as $C(J)\dd J\dd\Omega_x\dd\psi$ and integrating over $\S^{2}\times \S^1$, comparison with Eq.~\eqref{eq:induced-J} fixes $C(J)=J$. Hence
\begin{equation}
 \dd\Gamma_r^{\mathrm{ind}}
 =J\dd J\dd\psi\dd\Omega_x,
 \label{eq:local-J}
\end{equation}

Writing $J=r\sin\alpha$, where $\alpha$ is the incidence angle with the radial normal, gives $0\leq\alpha\leq\pi/2$. Using $\dd J=r\cos\alpha\dd\alpha$ and $\dd A_r=r^2\dd\Omega_x$, we obtain
\begin{equation}
 \dd\Gamma_r^{\mathrm{ind}}
 =\cos\alpha\,\dd A_r\,
 \sin\alpha\dd\alpha\dd\psi
 =\dd\Gamma_r^{\mathrm{intr}},
 \label{eq:sphereResult}
\end{equation}
which coincides with the intrinsic measure of $\mathcal{S}_r$.
The integrated check yields
\begin{equation}
 \int_0^r8\pi^2J\dd J=4\pi^2r^2
 =\pi\,\Area(\mathcal{S}_r).
 \label{eq:integrated-check}
\end{equation}
This verifies the Crofton normalization on every intermediate sphere. 

\subsection{Crofton reconstruction for a generic surface}\label{sec:crofton generic}

We now extend the sphere result to an arbitrary smooth surface in $\mathcal A$. At $P\in\mathcal A$, define the line density $\rho_{\ell}(P,[\mathbf v])$ as the weighted geodesic density per unit area transverse to $\mathbf v$ and per unit solid angle in the unoriented direction $[\mathbf v]$. Its dependence on $[\mathbf v]$ describes the angular distribution at fixed $P$, while its dependence on $P$ describes how that distribution varies across the shell. The boundary source measure and the multiplicity quotient fix the geodesic ensemble and its weights once and for all. The line density characterizes this ensemble, independently of the surface used to probe it. For an infinitesimal area element $\dd A_\Sigma$ at any $P$ with an arbitrary unit normal $\mathbf n$, the induced measure of its intersections with the boundary-sourced geodesics is
\begin{equation}
	\dd\Gamma_{P,\mathbf n}^{\mathrm{ind}}
	=\rho_{\ell}(P,[\mathbf v])\,
	\abs{\mathbf v\cdot\mathbf n}
	\dd A_\Sigma\,\dd\Omega_{[\mathbf v]}.
	\label{eq:arbitrary-element1}
\end{equation}

For the area element of $\mathcal S_r$ through $P$, the unit normal is $\mathbf e_{\hat r}$, so the same definition specializes to
\begin{equation}
 \dd\Gamma_{P,\mathbf e_{\hat r}}^{\mathrm{ind}}
 =\rho_{\ell}(P,[\mathbf v])
 \abs{\mathbf v\cdot\mathbf e_{\hat r}}
 \dd A_r\,\dd\Omega_{[\mathbf v]}.
 \label{eq:density-def}
\end{equation}
Comparison with Eq.~\eqref{eq:sphereResult} yields $\rho_{\ell}(P,[\mathbf v])=1$ for directions transverse to $\mathcal S_r$; the excluded tangential directions have zero angular measure. Because the sphere $\mathcal S_r$ and the crossing point $P$ were arbitrary, this conclusion holds throughout $\mathcal A$. At each point, the weighted geodesic density is isotropic in direction; its normalization is also the same at every point.  Hence, for an area element $\dd A_\Sigma$ at $P$ with an arbitrary normal $\mathbf n$, the induced measure of intersections coincides with the intrinsic measure:
\begin{equation}
 \dd\Gamma_{P,\mathbf n}^{\mathrm{ind}}
 =\abs{\mathbf v\cdot\mathbf n}
 \dd A_\Sigma\,\dd\Omega_{[\mathbf v]}
 =\dd\Gamma_{P,\mathbf n}^{\mathrm{intr}}.
 \label{eq:arbitrary-element}
\end{equation}
Integrating over all directions and over a smooth surface $\Sigma\subset\mathcal A$ therefore gives its area as a weighted intersection count: 
\begin{equation}
	\Area(\Sigma)=\frac{1}{\pi}
	\int_{\mathbb K_{\mathcal A}}N_\Sigma(\gamma)\,
	\dd\nu_{\mathcal A}(\gamma),
	\label{eq:crofton}
\end{equation}
Here $N_\Sigma(\gamma)$ counts the intersections of a chord $\gamma$ with $\Sigma$. Thus the fixed boundary-generated geodesic network reconstructs the area of any smooth surface in the shell. 

The equality in Eq.~\eqref{eq:arbitrary-element} is not imposed separately on each probe surface: the ensemble and its weights are fixed at $\S_+$. The calculation shows that its induced crossing measure equals the intrinsic measure of the probe surface. The result does not depend on the specific form of $\Lambda(r)$. The induced $J$-distribution on $\mathcal{S}_r$, Eq.~\eqref{eq:induced-J}, depends only on the boundary source distribution $\dd\Gamma_+^{(J)}=8\pi^2J\dd J$ and the multiplicity ratio $N_r(J)/N_+(J)$. Within the radial geometries considered here, the source measure and multiplicities depend only on the outer--outer/inner--outer chord classification, not on the detailed profile of $\Lambda(r)$. 

This universality of the launch measure should be distinguished from reconstructing an unknown bulk metric. At fixed $R_+$, different radial metrics have the same round boundary metric and the same launch measure, Eq.~\eqref{eq:source-J}. Their geodesic propagation is nevertheless different. For example, an outer--outer chord has angular advance
\begin{equation}
 \Delta\phi_{++}(J)=2\int_J^{R_+}
 \frac{J\Lambda(r)}{r\sqrt{r^2-J^2}}\,\dd r,
 \qquad R_-<J<R_+,
 \label{eq:review-angular-data}
\end{equation}
as follows from Eqs.~\eqref{eq:J-def} and~\eqref{eq:rdot}, and its length is given by Eq.~\eqref{eq:outer-outer-length}. Thus the endpoint relation, the angular advance including any winding, and the travel length can carry metric information absent from the launch density. The Crofton formula therefore represents areas on a supplied geometry using both source weights and geodesic propagation; it is not an inverse-metric reconstruction from the round boundary metric alone. An independent boundary formulation would need propagation-sensitive data and a proof of their sufficiency; no such inverse theorem is assumed here.
\subsection{Schwarzschild black hole versus vacuum AdS space}
\label{sec:homogeneous}

The comparison with vacuum AdS involves two distinct questions: how the geodesic measure is determined, and whether a chosen boundary source represents it completely. Symmetry simplifies the first question in vacuum AdS, whereas the Schwarzschild construction makes the second explicit.

A Riemannian manifold is homogeneous when its isometry group acts transitively on points. In the maximally symmetric examples---Euclidean space, the round sphere, and hyperbolic space, including a time-symmetric spatial slice of vacuum AdS---the isometries also relate all unoriented geodesics. The geodesic space can then be described as the quotient of the isometry group by the stabilizer of a reference geodesic. The invariant group measure induces a kinematic measure on this quotient, with its overall normalization fixed by the Crofton convention \cite{Santalo,Helgason}. Thus the measure can be specified without first selecting a boundary source.\footnote{In vacuum AdS, the kinematic space was shown to be a coadjoint orbit of the conformal group, whose Crofton form coincides with the Kirillov--Kostant symplectic form on that orbit \cite{PennaZukowski}.}

By contrast, the time-symmetric Schwarzschild slice has no isometries relating different areal radii. Its continuous isometries form $SO(3)$, which acts transitively on each sphere $\mathcal S_r$ but not on the full slice. Rotational invariance therefore does not determine the complete shell-chord measure. The boundary-generated description instead requires geodesic transport, complete directional coverage, and the source-multiplicity quotient established above.

 Preserving the weight of the sourced geodesics does not ensure that every required direction is present. The next section tests this distinction on the Einstein--Rosen bridge; the general Riemannian formulation of the transport argument is developed in Section~\ref{sec:liouville transport}.

\section{Crofton reconstruction on the Einstein--Rosen bridge}
\label{sec:bridge}

Section~\ref{sec:2} considered a shell entirely outside the Schwarzschild horizon, with $R_->r_h$. We now ask whether the same boundary-sourced construction extends through the bridge throat. This is not a continuation into the Lorentzian black-hole interior: for $r<r_h$, the roles of $t$ and $r$ interchange, and a constant exterior-time slice does not remain spacelike. Instead, we consider the time-symmetric Einstein--Rosen slice of the maximally extended spacetime, in analogy with the eternal black hole in AdS \cite{Maldacena:2001kr}.  On this slice $r\ge r_h$ everywhere, and the geometry consists of two isometric asymptotically flat regions joined through a throat at $r=r_h$. We place $\S_+$ on the right and $\S_-$ on the left, at areal radii $R_+>r_h$ and $R_->r_h$, respectively, and study the bridge region bounded by these spheres.

\subsection{Spatial geometry and geodesics on the bridge slice}

On each exterior portion of this slice, the spatial metric is
\begin{equation}
	\dd s^2 = \Bigl(1-\frac{r_h}{r}\Bigr)^{-1}\!\dd r^2 + r^2\dd\Omega_2^2,
	\qquad r\ge r_h\ \text{on both sides},
	\label{eq:bridge-t-const}
\end{equation}
The apparent singularity at $r=r_h$ is removed by a signed proper radial coordinate $\rho$, chosen positive on the right and negative on the left, with $|\dd\rho|=|\dd r|/\sqrt{1-r_h/r}$. In this coordinate the metric reads
\begin{equation}
	\dd s^2 = \dd\rho^2 + r(\rho)^2\,\dd\Omega_2^2,
	\qquad
	\left(\frac{\dd r}{\dd\rho}\right)^2 = 1 - \frac{r_h}{r},
	\label{eq:bridge-proper}
\end{equation}
with the throat at $\rho=0$, $r(0)=r_h$.  Near the throat,
$r(\rho) = r_h + \rho^2/(4r_h) + \cO(\rho^4)$, so $r(\rho)$ is even in $\rho$
and attains its global minimum at the throat.  

The sphere at $r$ has area $4\pi r^2$, so $r$ is the \emph{areal radius}, whereas $\rho$ measures signed \emph{proper radial distance} from the throat. Unlike $r$, $\rho$ distinguishes the two sides of the bridge. For a geodesic that crosses the throat, the areal radius $r$ is \emph{not monotonic} in arclength~$s$: it decreases from the source value $R_+$ on the right side down to its minimum $r_h$ at the throat, then \emph{increases} from $r_h$ to $R_-$ on the left side. The coordinate ranges are summarized below:
\begin{center}
	\begin{tabular}{c c c}
		\toprule
		& areal radius $r$ & proper radial distance $\rho$\\
		\midrule
		right asymptotic boundary & $r\to\infty$ & $\rho\to+\infty$\\
		throat & $r=r_h$ & $\rho=0$\\
		left asymptotic boundary & $r\to\infty$ & $\rho\to-\infty$\\
		\bottomrule
	\end{tabular}
\end{center}
Thus $r$ no longer labels nested spheres monotonically across the full region, as assumed in Section~\ref{sec:transport}. We must therefore re-examine the geodesic sectors before extending that reconstruction.

A unit-speed geodesic on the bridge obeys
\begin{equation}
	\dot\rho^2 + \frac{J^2}{r(\rho)^2} = 1,
	\qquad J = \abs{r^2\dot\phi} = \text{const},
	\label{eq:bridge-eom}
\end{equation}
with $r(\rho)$ given by Eq.~\eqref{eq:bridge-proper}.  Consider a geodesic
launched from the right outer sphere.  Since $\dot\rho^2 = 1 - J^2/r^2$, a
turning point requires $r=J$.  As $r(\rho)\ge r_h$ everywhere, geodesics with
$J>r_h$ can never reach the throat, for $\dot\rho^2$ would become negative at
$r<J$. They instead turn around at $r=J>r_h$ on the right side.  Geodesics with
$J<r_h$ have $\dot\rho^2>0$ everywhere, since $r(\rho)\ge r_h>J$ on \emph{both}
sides of the bridge, and so cross the throat and continue indefinitely into the
left asymptotic region.  The critical geodesic with $J=r_h$ has $\dot\rho^2=0$
at the throat and spirals onto it with logarithmically divergent arclength.
This set has zero Liouville measure and will be ignored.

For the noncritical sectors launched from the right sphere $\S_{+}(R_+)$, the crossing multiplicities are therefore:
\begin{enumerate}
	\item \emph{Crossing geodesics} ($0 < J < r_h$).  These travel from
	$\S_{+}$ on the right side, pass through the throat, and emerge on the
	left side, reaching the left boundary sphere $\S_{-}(R_-)$.
	They intersect $\S_{+}$ once, the throat once, and $\S_{-}$ once.
	Hence $N_+ = 1$, $N_{r_h} = 1$, $N_- = 1$.
	
	\item \emph{Right-returning geodesics} ($r_h < J < R_+$).
	These turn around at $r = J$ on the right side before reaching the
	throat, and return to $\S_{+}$.  Hence $N_+ = 2$, and they never
	intersect the left side.
\end{enumerate}
Figure~\ref{fig:einstein_rosen_bridge} illustrates these sectors. By reflection symmetry, a source on the left sphere generates crossing geodesics with $0<J<r_h$ and left-returning geodesics with $r_h<J<R_-$. The crossing sector is shared by the two sources as a set of unoriented bridge chords, whereas their returning sectors are distinct. 

\begin{figure}[h]
	\centering
	\includegraphics[scale=1]{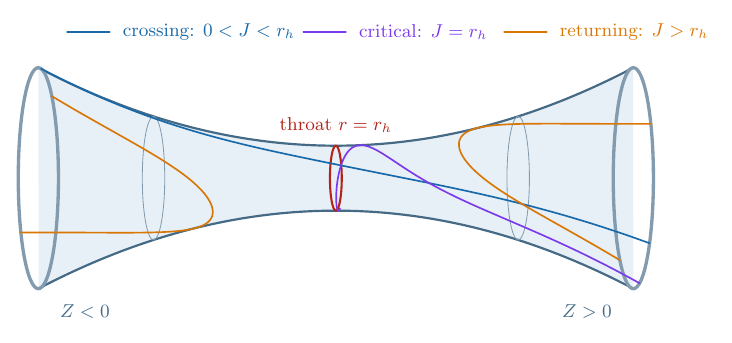}
	\caption{The Einstein--Rosen bridge, depicted as a Flamm paraboloid embedding. Geodesics launched from the right outer sphere $\S_{+}$ split into crossing geodesics (blue curves, $0<J<r_h$), which pass through the throat and reach the left side, and returning geodesics (orange curves, $r_h<J<R_+$), which turn around before the throat and return to $\S_{+}$. The sectors launched from the left outer sphere are mirror images of these. The purple critical geodesic with $J=r_h$ has zero measure.}	
	\label{fig:einstein_rosen_bridge}	
\end{figure}

Since each crossing chord intersects the throat once, the measure of the crossing sector directly determines its Crofton area count:
\begin{equation}
	\nu_{\mathcal A}(\cK_{\mathrm{cross}})
	=8\pi^2\int_0^{r_h}J\,\dd J
	=4\pi^2 r_h^{\,2},
	\label{eq:cross-nu}
\end{equation}
which, after multiplication by $1/\pi$, reproduces the horizon area
$A_H\equiv \Area(\S_{r_h})=4\pi r_h^{\,2}$.

The interpretation of this horizon-area count in terms of entanglement entropy, Bekenstein--Hawking entropy, and  the Einstein--Rosen (ER) = Einstein--Podolsky--Rosen (EPR) proposal is a physical statement that goes beyond the geometric reconstruction presented here. We defer it to Section~\ref{sec:pee-shell}.

\subsection{Failure of the Crofton formula and the need for bilateral sources}

We now ask: from the \emph{single right-side source} alone, what is the
transported measure on a sphere $\mathcal{S}_r$ located on the \emph{left} side
of the bridge, at areal radius $r_h<r<R_-$?

\begin{itemize}
	\item Crossing geodesics ($0 < J < r_h$): every one of these traverses
	the left side and intersects $\mathcal{S}_r$ exactly once.  $N_r = 1$.
	\item Right-returning geodesics ($J > r_h$): none of these cross the
	throat.  $N_r = 0$.
\end{itemize}
The intersections of $\mathcal S_r$ with the right-sourced geodesics therefore inherit the following induced $J$-marginal:
\begin{equation}
	\dd\Gamma_r^{\mathrm{ind},(J)}
	= \frac{1}{1}\; 8\pi^2 J\,\dd J
	= 8\pi^2 J\,\dd J,
	\qquad 0 < J < r_h,
	\label{eq:induced-left}
\end{equation}
and zero for $J \ge r_h$.  The total weighted intersection count is
\begin{equation}
	\int_0^{r_h} 8\pi^2 J\,\dd J
	= 4\pi^2 r_h^{\,2}
	= \pi\,A_H,
	\label{eq:total-weight-left}
\end{equation}
This is smaller than $\pi\Area(\mathcal S_r)$ for every left-side sphere with $r>r_h$. The right-sourced measure continues through the throat, but it no longer supplies the full crossing measure needed for area reconstruction.

At the throat, by contrast, Eq.~\eqref{eq:total-weight-left} agrees with the intrinsic area count $\pi A_H$. The single right-side source therefore remains sufficient there, even though it fails on every sphere strictly to the left.

The throat is also geometrically distinguished as a \emph{totally geodesic} submanifold: 
its extrinsic curvature vanishes identically (not merely its trace),
so any geodesic tangent to the throat remains confined to it.
In the regular proper coordinate $\rho$, its unit normals are simply
$\pm\partial_\rho$. Reversing the normal changes the sign of the oriented
crossing form, but leaves the positive unoriented measure unchanged.
The intrinsic crossing measure therefore remains well-defined at the
throat. It is the radial coordinate $r$, rather than that measure, which
degenerates there.  The Schwarzschild
relation $r_h=2Gm$ then gives $A_H=16\pi G^2m^2$.

The missing directional coverage can be quantified locally. At a point on the left sphere, let $\alpha$ be the acute angle of an unoriented line with the radial normal. Since $J=r\sin\alpha$, the right source supplies only
\begin{equation}
 \sin\alpha<\frac{r_h}{r},\qquad
 \frac{\displaystyle\int_0^{\arcsin(r_h/r)}
     \cos\alpha\sin\alpha\,\dd\alpha}
 {\displaystyle\int_0^{\pi/2}\cos\alpha\sin\alpha\,\dd\alpha}
 =\frac{r_h^2}{r^2}.
 \label{eq:review-bridge-coverage}
\end{equation}
This is a fraction of the normal flux, not the solid-angle fraction. The transported gas reaches every point on these spheres, but leaves out the more tangential directions.  At the throat the allowed cone contains all unoriented directions up to a null set; on the left it narrows as $r$ increases. Rescaling the total count to fit one sphere would not fill the missing directions or restore the intrinsic measure for arbitrary surface normals. The failure is therefore a directional obstruction, not merely an overall normalization error.
The missing directions belong to the left-returning sector, which a right-side source cannot generate. A complete reconstruction therefore requires bilateral sourcing:

\begin{itemize}
	\item For a complete Crofton reconstruction on \emph{both} sides of the
	bridge, one requires sources on \emph{both} boundary
	spheres.  
	\item The source on $\S_-(R_-)$ supplies the missing range $r_h<J<r$ on each left-side sphere with $r_h<r<R_-$ through its returning sector.  
	\item The two sources share the crossing sector, whose total unoriented measure is 
	$4\pi^2 r_h^{\,2} = \pi A_H$. 
\end{itemize}

Bilateral sourcing must include a single quotient over the combined source. If $N_{\mathrm R}(\gamma)$ and $N_{\mathrm L}(\gamma)$ count source representatives of a bridge chord, the analogue of Eq.~\eqref{eq:quotient} uses $N_{\mathrm R}(\gamma)+N_{\mathrm L}(\gamma)$. A crossing chord has one representative on each side and is counted once after division by two; a returning chord has two on its own side and is likewise counted once. Simply adding the two already-quotiented single-side gases would double the shared crossing sector. With this common quotient, the two sources supply the full intrinsic crossing measure throughout the region between the boundary spheres. Their union thus provides a perfect tessellation of this bridge region. In the factorized EPR-pair model of Section~\ref{sec:pee-shell}, the shared crossing sector carries the left--right entanglement, while returning chords have both endpoints on the same side. This geometric statement does not by itself establish a quantum reconstruction algebra or identify a causal information barrier.
\section{Liouville transport and Crofton reconstruction on more generic Riemannian manifolds}
\label{sec:liouville transport}

We now formulate the transport argument used in the shell and bridge examples. The aim is to separate preservation of the local crossing measure from completeness of a chosen source. The argument has three steps.
\begin{itemize}
	\item We recall the Hamiltonian description of geodesic flow and Liouville invariance.
	
	\item We condition the invariant measure on unit speed, extract its transverse flux, and prove its preservation between local crossing sections. 
	
	\item We apply the result to the source on $\S_+$ in Section~\ref{sec:2}, using directional reachability and the source-multiplicity quotient to recover the full crossing measure in the shell.
\end{itemize}

The first two steps use only a supplied smooth Riemannian metric; they do not require homogeneity or the Einstein equations. The last step is source-dependent and must be checked globally. The bridge provides a concrete illustration: incompleteness of a single source leaves the local transport law intact but exposes missing geodesic sectors. The generality of the geometric description and the completeness of one boundary realization are therefore distinct claims.

\subsection{Hamiltonian phase space and Liouville invariance}\label{sec:phase space}

Let $\M$ be an $n$-dimensional smooth Riemannian manifold with metric $h$, whose phase space is $T^*\M$ with canonical coordinates $\xi=(x,p)=(x^i,p_i)$.  The canonical one-form, the symplectic form, and the Liouville measure are
\begin{equation}
	\vartheta=p_i\,{\dd}x^i,
	\quad
	\omega=\dd\vartheta
	={\dd}p_i\wedge{\dd}x^i,
	\quad
	\dd\mu=\frac{1}{n!}\omega^{\wedge n}.
	\label{eq:canonical phase space structures}
\end{equation}
The last equality fixes the orientation convention used below.  Reversing this convention changes only the sign of oriented forms, not the associated positive measures.

For a smooth Hamiltonian $H(x,p)$ with no explicit dependence on the flow parameter, define its Hamiltonian vector field $X_H$ by
\begin{equation}
	\iota_{X_H}\omega=-{\dd}H,
	\label{eq:Hamiltonian vector field definition}
\end{equation}
where $\iota_X\omega$ denotes the contraction of the first slot of the symplectic form $\omega$ with the vector field $X$, namely $(\iota_X\omega)(Y)=\omega(X,Y)=\omega_{ij}X^iY^j$.  Since $\{\partial_{x^i},\partial_{p_i}\}$ is a coordinate basis for the phase space, $X_H$ has a unique local expansion
\begin{equation}
	X_H=A^i\partial_{x^i}+B_i\partial_{p_i},
\end{equation}
where $A^i$ and $B_i$ can be deduced by substitution into the defining equation \eqref{eq:Hamiltonian vector field definition}:
\begin{equation}
	\iota_{X_H}\omega=B_i\,{\dd}x^i
	-A^i\,{\dd}p_i
	=-\frac{\partial H}{\partial x^i}{\dd}x^i
	-\frac{\partial H}{\partial p_i}{\dd}p_i,
\end{equation}
hence
\begin{equation}
		X_H
	=\frac{\partial H}{\partial p_i}\partial_{x^i}
	-\frac{\partial H}{\partial x^i}\partial_{p_i}.
	\label{eq:Hamilton equations on phase space}
\end{equation}

The Hamiltonian vector field $X_H$ generates the local Hamiltonian flow $\varphi_H^s$ on $T^*\M$.  For every initial point $\xi\in T^*\M$, and for all parameter values for which the solution remains in the chosen open phase-space domain, we have
\begin{equation}
	\frac{\dd}{{\dd}s}\varphi_H^s(\xi)=X_H\bigl(\varphi_H^s(\xi)\bigr),
	\quad
	\varphi_H^0(\xi)=\xi,
	\quad
	\varphi_H^{s+t}=\varphi_H^s\circ \varphi_H^t,
	\label{eq:Hamiltonian flow definition}
\end{equation}
and the corresponding Lie derivative is
\begin{equation}
	\mathscr{L}_{X_H}\eta
	\equiv \left.\frac{\dd}{{\dd}s}\right|_{s=0}(\varphi_H^s)^*\eta,
	\label{eq:Lie derivative from Hamiltonian flow}
\end{equation}
for every differential form $\eta$, where the superscript $*$ denotes pullback.  Consequently, $\mathscr{L}_{X_H}\eta=0$ is equivalent to invariance $(\varphi_H^s)^*\eta=\eta$ along the flow. For a scalar function $\mathfrak f$, the pullback means composition, and hence
\begin{equation}
	(\mathscr{L}_{X_H}\mathfrak f)(\xi)=\left.\frac{\dd}{{\dd}s}\right|_{s=0}\left[(\varphi_H^s)^*\mathfrak f\right](\xi)=\left.\frac{\dd}{{\dd}s}\right|_{s=0}\mathfrak f\bigl(\varphi_H^s(\xi)\bigr)=X_H(\mathfrak f)(\xi)
	.
\end{equation}
This implies that the Hamiltonian itself is invariant under the Hamiltonian flow as expected:
\begin{equation}
	\mathscr{L}_{X_H}H=X_H(H)=\frac{\partial H}{\partial p_i}\partial_{x^i}
	H-\frac{\partial H}{\partial x^i}\partial_{p_i}H=0\quad \Rightarrow \quad H\circ \varphi_H^s=H.
	\label{eq:Hamiltonian conservation}
\end{equation}
There are other important quantities which are invariant along the flow. We first recall Cartan's formula and the Leibniz rule for a vector field $X$ and differential forms $\eta,\zeta$:
\begin{equation}
	\mathscr{L}_X\eta
	=\dd(\iota_X\eta)+\iota_X(\dd\eta),
	\label{eq:Cartan}
\end{equation}
\begin{equation}
	\mathscr{L}_X(\eta\wedge\zeta)
	=(\mathscr{L}_X\eta)\wedge\zeta
	+\eta\wedge\mathscr{L}_X\zeta.
	\label{eq:Leibniz}
\end{equation}
These identities imply that the closed symplectic form, $\dd\omega=0$, is invariant:
\begin{equation}
	\mathscr{L}_{X_H}\omega
	=\dd(\iota_{X_H}\omega)
	+\iota_{X_H}(\dd\omega)
	=\dd(-{\dd}H)+0=0.
\end{equation}
Hence the Liouville measure is invariant by the Leibniz rule:
\begin{equation}\label{eq:Liouville thm}
	\mathscr{L}_{X_H}\dd\mu
	=\frac{1}{n!}\mathscr{L}_{X_H}\left(\omega^{\wedge n}\right)
	=\frac{1}{(n-1)!}(\mathscr{L}_{X_H}\omega)\wedge\omega^{\wedge(n-1)}
	=0,
\end{equation}
which is Liouville's theorem in the present sign convention.  This conclusion follows from the Hamiltonian structure and does not depend on the explicit form of the metric $h$.  

\subsection{Transporting the measure along geodesics}\label{sec:intrinsic measure}

Equation~\eqref{eq:Hamiltonian conservation} shows that each energy level is preserved by the flow. For geodesic motion, the Hamiltonian is
\begin{equation}
	H(x,p)=\frac{1}{2}h^{ij}p_{i}p_{j}=\frac{1}{2}h_{ij}\dot{x}^{i}\dot{x}^{j},
\end{equation}
which is obtained from the geodesic Lagrangian $L=\frac{1}{2}h_{ij}\dot{x}^{i}\dot{x}^{j}$ through the Legendre relation $p_i=h_{ij}\dot{x}^j$, while the corresponding Euler--Lagrange equations are the geodesic equations.  Fixing the value of $H$ fixes the multiplicative normalization of an affine parameter, although an additive shift remains free. For example, when choosing the affine parameter to be the arclength $s$, defined by ${\dd}s^2=h_{ij}{\dd}x^i{\dd}x^j$, where $x^i$ are the spatial coordinate functions, we have
\begin{equation}
	H(x,p)=\frac12h^{ij}p_ip_j
	=\frac{1}{2}h_{ij}
	\frac{{\dd}x^i}{{\dd}s}
	\frac{{\dd}x^j}{{\dd}s}=\frac12,
\end{equation}
which we call the unit-speed condition. We thus call the surface with $H=1/2$ in the phase space the unit-speed surface:
\begin{equation}
	\Sbb_{H=1/2}
	=\left\{(x,p)\in T^*\M:
	H(x,p)=\frac12\right\}.
\end{equation}

We first condition the Liouville measure on this unit-speed surface, then contract with the flow to obtain a transverse crossing measure. This conditioning is not the pullback of the ambient top-degree form, which would vanish on the energy surface. Here and below, a geodesic intersects a surface $\Sigma$ transversely at some point when its tangent vector $\mathbf v$ is not tangent to $\Sigma$, or equivalently $h(\mathbf v,\mathbf n_\Sigma)\neq 0$ for a unit normal $\mathbf n_\Sigma$.  The excluded case $h(\mathbf v,\mathbf n_\Sigma)=0$ is a tangent intersection.

Since $p\ne0$ on $H=1/2$, this is a regular energy level: $\dd H\ne0$, and $X_H$ has no zero there. For a smooth compactly supported test function $F$ on this level, define its microcanonical measure $\dd\lambda$ by
\begin{equation}
	\int_{\Sbb_{H=1/2}}F\,\dd\lambda
	\equiv \int_{T^*\M}
	F\,\delta\left(H-\frac12\right)\dd\mu,
	\label{eq:microcanonical definition}
\end{equation}
where $\delta$ is the Dirac distribution, and only the values of $F$ on the energy level contribute. The sign convention places the energy coordinate last: locally $\dd\mu=\widetilde{\dd\lambda}\wedge\dd H$, where the tilde denotes an ambient representative of the energy-level form. Thus $\dd\lambda$ is a smooth nonvanishing volume form on the level; its associated positive density is used in integrals.\footnote{Strictly speaking, the $F$ on the right-hand side denotes any smooth ambient extension $\widetilde F$ of the test function from $\Sbb_{H=1/2}$ to a neighborhood of that surface in $T^*\M$.  The result is independent of the chosen extension because the Dirac distribution is supported on $H=1/2$.} Using \eqref{eq:Hamiltonian conservation} and 	\eqref{eq:Liouville thm}, and the product rule, we find that the distributional representative is invariant:
\begin{equation}
	\begin{aligned}
		\mathscr{L}_{X_H}\left[\delta\left(H-\frac12\right)\dd\mu\right]
		&=\mathscr{L}_{X_H}\left[\delta\left(H-\frac12\right)\right]\dd\mu
		+\delta\left(H-\frac12\right)\mathscr{L}_{X_H}\dd\mu\\
		&=\delta'\left(H-\frac12\right)X_H(H)\dd\mu
		+\delta\left(H-\frac12\right)\mathscr{L}_{X_H}\dd\mu\\
		&=0.
	\end{aligned}
\end{equation}
Thus the induced measure $\dd\lambda$ is invariant under the Hamiltonian flow $\varphi_H^s$ on the unit-speed surface $\Sbb_{H=1/2}$:
\begin{equation}
	\mathscr{L}_{X_H}\dd\lambda=0,
	\label{eq:Liouville invariance dlambda}
\end{equation}
or equivalently, $(\varphi_H^s)^*\dd\lambda=\dd\lambda$. 

We next determine the local form of $\dd\lambda$ and remove the arclength direction to obtain a transverse measure over a surface in the base manifold $\M$. Let
\begin{equation}
	\pi:\Sbb_{H=1/2}\longrightarrow\M,
	\quad
	\pi(x,p)=x,
\end{equation}
be the bundle projection.  For a smooth embedded two-sided hypersurface $\Sigma\subset\M$, the oriented phase-space crossing section is
\begin{equation}
	C_\Sigma
	\equiv \left\{(x,p)\in\Sbb_{H=1/2}:
	x\in\Sigma,\
	h(\mathbf v,\mathbf n_\Sigma)\neq0,\
	v^i=h^{ij}p_j\right\},
	\label{eq:phase-space section over Sigma}
\end{equation}
which encodes the local information of all unit vectors that intersect $\Sigma$ transversely. Equivalently, $C_\Sigma$ is the transverse part of $\pi^{-1}(\Sigma)\subset\Sbb_{H=1/2}$. 

We now specialize to three dimensions, which is the spatial dimension relevant to this paper, and the construction generalizes directly to higher dimensions. Fix $P\in\Sigma$ and choose local coordinates $(y^1,y^2,z)$ adapted to $\Sigma$, with $z$ the signed proper normal distance\footnote{Here ``signed'' means that $z>0$ in the direction of $\mathbf n_\Sigma$ and $z<0$ in the opposite direction.  } and the coordinate frame orthonormal at $P$.  Let ${\dd}V_h$ denote the volume form and $\dd A_\Sigma$ the local area measure at $P$, and write the covector momentum in the associated orthonormal coframe as $p=\bar p_a e^a$.  The Liouville measure in \eqref{eq:canonical phase space structures} at $P$ takes the local form
\begin{equation}
	\dd\mu
	={\dd}V_h\wedge\dd^{3}\bar p
	={\dd}z\wedge\dd A_\Sigma
	\wedge\dd^{3}\bar p,
\end{equation}
where $\dd^{3}\bar p=\dd\bar p_1\wedge\dd\bar p_2\wedge\dd\bar p_3$.  Write $\kappa=|\bar p|$ and $\mathbf v=\frac{\bar p}{|\bar p|}\in \S^{2}$. Let $\dd\Omega_{\v{v}}$ denote the standard solid-angle form of the oriented unit vector $\bar p/\kappa$.  Since $\dd^{3}\bar p=\kappa^2\dd\kappa\wedge\dd\Omega_{\v{v}}$ at $P$, the definition \eqref{eq:microcanonical definition} of the induced measure gives
\begin{equation}
	\begin{aligned}
		\dd\lambda
		&={\dd}z\wedge\dd A_\Sigma
		\wedge\dd\Omega_{\v{v}}
		\int_0^\infty
		\delta\left(\frac{\kappa^2-1}{2}\right)\kappa^2\dd\kappa\\
		&={\dd}z\wedge\dd A_\Sigma
		\wedge\dd\Omega_{\v{v}}.
	\end{aligned}
\end{equation}
Here we used $\delta(\zeta(\kappa))=\sum_k\delta(\kappa-\kappa_k)/|\zeta'(\kappa_k)|$ for the simple zeros $\kappa_k$ of $\zeta$.

The induced measure $\dd\lambda$ still contains the arclength direction along a geodesic-flow orbit. For an infinitesimal flow displacement ${\dd}s$, the signed normal displacement obeys
\[
 \dd z=h(\mathbf v,\mathbf n_\Sigma)\dd s
       =\cos\alpha_\Sigma\,\dd s.
\]
We therefore define the oriented transverse measure
\begin{equation}
	\beta\equiv \iota_{X_H}\dd\lambda.
	\label{eq:oriented transverse measure}
\end{equation}
The form $\beta$ is horizontal with respect to the flow, since $\iota_{X_H}\beta=0$. Together with flow invariance, this makes it a local transverse measure; no smooth global quotient of the geodesic flow is assumed. Indeed, $X_H(z)={\dd}z/{\dd}s=h(\mathbf v,\mathbf n_\Sigma)$, so the pullback of the contraction to the oriented phase-space crossing section $C_\Sigma$ is
\begin{equation}
	\left.\beta\right|_{C_\Sigma}=h(\mathbf v,\mathbf n_\Sigma)\,
	\dd A_\Sigma\wedge\dd\Omega_{\v{v}}.
\end{equation}
Thus $\beta$ is an oriented transverse measure whose pullback gives a signed crossing density. The unoriented crossing measure needed is the positive measure associated with its absolute density.  Equivalently, one may choose $h(\mathbf v,\mathbf n_\Sigma)>0$ or take the absolute density on the quotient $[\mathbf v]\in\mathbb{RP}^2=\Sbb^2/(\mathbf v\sim-\mathbf v)$.  Hence
\begin{equation}\label{eq:intrinsic measure}
	\dd\Gamma_\Sigma
	=\abs{\beta}_{[\v{v}]}=\abs{h(\mathbf v,\mathbf n_\Sigma)}\,
	\dd A_\Sigma\,\dd\Omega_{[\v{v}]}
	=\abs{\cos\alpha_\Sigma}\,
	\dd A_\Sigma\,\dd\Omega_{[\v{v}]},
\end{equation}
which is the intrinsic measure, also called the crossing measure, for the geodesics crossing $\Sigma$.  Here $\dd\Omega_{[\v{v}]}$ is normalized by one hemispherical representative of each unoriented line, equivalently by one half of the push-forward of spherical solid angle. Thus passing from oriented vectors to unoriented lines introduces no extra factor of two in Eq.~\eqref{eq:intrinsic measure}. On a local branch, one chooses an oriented lift and then takes the positive density. Using a hemispherical representative, its normalization is
\begin{equation}\label{eq:uniform normalization}
	\int_{\mathbb{RP}^2}\dd\Gamma_\Sigma
	=\dd A_\Sigma\int_{\mathbb{RP}^2}
	\abs{\cos\alpha_\Sigma}\,\dd\Omega_{[\v{v}]}
	=\pi\,\dd A_\Sigma.
\end{equation}
The last equality is the standard hemisphere integral. Consequently, the direction-integrated crossing density per unit area is the universal constant $\pi$. For dimension $n\ge2$, the corresponding angular integral on $\mathbb{RP}^{n-1}$ is $c_n=\operatorname{Vol}(\mathbb S^{n-2})/(n-1)$, the volume of the Euclidean unit $(n-1)$-ball. Thus the subsequent formulas with $\pi$ and $\mathbb{RP}^2$ refer to three spatial dimensions; in higher dimensions they use $c_n$ and $\mathbb{RP}^{n-1}$.

Now consider two transverse phase-space crossing sections $C_{1,2}\subset\Sbb_{H=1/2}\subset T^*\M$ with $\pi(C_{1,2})=\Sigma_{1,2}$, and let $\tau(\xi)$ denote the travel arclength required for the geodesic-flow orbit starting at $\xi\in C_1$ to reach a chosen intersection with $C_2$. The function $\tau$ need not be defined on all of $C_1$, since some geodesic-flow orbits starting in $C_1$ may not reach $C_2$. Near any existing transverse intersection reached in finite arclength, however, the implicit function theorem allows us to choose sufficiently small open subsets $C'_{1,2}\subset C_{1,2}$ on which the chosen travel-arclength branch $\tau$ is smooth and well-defined. We therefore define the transporting map $P_{12}:C'_1\to C'_2$ by
\begin{equation}\label{eq:transporting map}
	P_{12}(\xi)=\varphi_H^{\tau(\xi)}(\xi),\quad \xi\in C'_1.
\end{equation}
In the base manifold $\M$, this map follows a geodesic segment from $\pi(\xi)\in\Sigma_{1}$ to $\pi(P_{12}(\xi))\in\Sigma_{2}$. The initial state $\xi$ and the chosen intersection determine this segment uniquely. The segment is not tangent to either $\Sigma_{1}$ or $\Sigma_{2}$ because both phase-space crossing sections are assumed transverse. The main result of this subsection is a transport theorem, stated and proved in Appendix~\ref{sec:app-transport}: under the hypotheses specified there, the oriented transverse measure $\beta$ is invariant under the transporting map $P_{12}(\xi)=\varphi_H^{\tau(\xi)}(\xi)$, $\xi\in C'_1$,
\begin{equation}\label{eq:transporting map 1}
	P_{12}^*\left(\left.\beta\right|_{C_2'}\right)
	=\left.\beta\right|_{C_1'}.
\end{equation}

We now discuss the geometric meaning of \eqref{eq:transporting map 1}. It states that the oriented transverse measure is preserved on every smooth geodesic-flow map connecting two transverse crossing sections. Since the arclength direction has already been removed in $\beta$ and $\beta$ is invariant under the fixed-arclength flow, varying the travel arclength across the branch does not change this measure. Furthermore, taking absolute densities shows that the intrinsic measure of the form \eqref{eq:intrinsic measure} is invariant along the geodesic flow, i.e. a measure initially given by \eqref{eq:intrinsic measure} retains this intrinsic form on the reached crossing section under the transporting map \eqref{eq:transporting map}. It is useful to express the theorem in terms of the intersections. Let $\Sigma_{1,2}'=\pi(C_{1,2}')$ denote the open subsets of $\Sigma_{1,2}$ over which the crossing sections $C_{1,2}'$ lie. A geodesic gas launched from $C_1'$ over $\Sigma_1'$, with this intrinsic source measure, propagates along the geodesics and arrives at $C_2'$ over $\Sigma_2'$. The induced crossing measure for the intersections on $\Sigma_2'$, with directions restricted to the reached states in $C_2'$, still takes the form of the intrinsic measure.
	
To verify this independently, we first write the intrinsic measures on $C'_1$ as follows
\begin{equation}\label{eq:intrinsic form 1}
	\left.\beta\right|_{C'_1}=\cos\alpha_1\,
	\dd A_{\Sigma_1}\,\dd\Omega_{\v{v}_1},
	\quad \dd\Gamma_1
	=\abs{\left.\beta\right|_{C'_1}}_{[\v{v}_1]}=\abs{\cos\alpha_1}\,
	\dd A_{\Sigma_1}\,\dd\Omega_{[\v{v}_1]}.
\end{equation}
Here $\sigma_1$ denotes local coordinates on the spatial surface $\Sigma_1$,
and $(\alpha_1,\psi_1)$ are local spherical coordinates for an oriented
unit direction at $\Sigma_1$. We label the coordinates on $C'_{2}$ by replacing the subscript 1 by 2 in this notation, and assume the two chosen sections are connected via \eqref{eq:transporting map}. We further use $\mathbf u_1$ and $\mathbf u_2$ to collectively denote the local coordinates on $C'_1$ and $C'_2$, respectively, such that \eqref{eq:transporting map} gives
\begin{equation}
	\mathbf u_2=P_{12}(\mathbf u_1)\quad \text{for}\quad \mathbf u_2=(\sigma_2,\alpha_2,\psi_2)\text{ on }C_2',\quad \mathbf u_1=(\sigma_1,\alpha_1,\psi_1)\text{ on }C_1'.
\end{equation}
In local coordinates, \eqref{eq:transporting map 1} has the following change-of-variables expression for the associated positive measures: the transported density on $C'_2$ equals the source density on $C'_1$ multiplied by $\mathcal J_{12}^{-1}$, where $\mathcal J_{12}\equiv
\abs{\det\!\left(\frac{\partial \mathbf u_2}{\partial \mathbf u_1}\right)}$ is the Jacobian between $\mathbf u_1$ and $\mathbf u_2$. For the oriented form $\beta$, the corresponding determinant is signed rather than absolute. In this calculation, we do not assume that the transported measure equals the intrinsic measure on $C'_2$. Instead, we would like to see whether multiplying the source density in \eqref{eq:intrinsic form 1} by the Jacobian $\mathcal J_{12}^{-1}$ can yield the same intrinsic form on $C'_2$. The calculation detailed in the footnote gives\footnote{{One strategy is to first re-express \eqref{eq:intrinsic form 1} in terms of explicit coordinates in $\mathbf u_1$, which gives
	\begin{equation}\label{eq:rho1}
		\dd\Gamma_{1}=\mathcal D_1(\mathbf u_1){\dd}\mathbf u_1,\quad \mathcal D_1(\mathbf u_1)=\abs{\cos\alpha_1}\sqrt{\det h_{\Sigma_1}(\mathbf u_1)}\sin\alpha_1,
	\end{equation} 
	where $h_{\Sigma_1}$ is the induced metric on $\Sigma_1$. From the passive viewpoint of \eqref{eq:transporting map}, the transported measure on $C'_{2}$ is given by $\dd\Gamma_{2}^{\text{trans}}\equiv 
	\left[\mathcal D_1(\mathbf u_1)\mathcal J_{12}^{-1}\right]{\dd}\mathbf u_2$. To compute it, we next introduce two auxiliary coordinate systems $z_{1,2}$ on $C'_{1,2}$ with canonical tangential momentum, $z_i=(\sigma_i^1,\sigma_i^2,p_{i1}^{\mathrm{tan}},p_{i2}^{\mathrm{tan}})$, where $p_{ia}^{\mathrm{tan}}$ are dual to the surface coordinates $\sigma_i^a$. The orthonormal tangential momentum components satisfy $\bar p_{i1}=\sin\alpha_i\cos\psi_i, \bar p_{i2}=\sin\alpha_i\sin\psi_i$. On a fixed transverse hemisphere, converting these components to the canonical ones contributes $\sqrt{\det h_{\Sigma_i}}$. Away from angular-coordinate singularities, we then have
	\begin{align}
		\abs{\det\!\left(
			\frac{\partial(\bar p_{i1},\bar p_{i2})}
			{\partial(\alpha_i,\psi_i)}\right)}
		=\abs{\cos\alpha_i}\sin\alpha_i\quad \Rightarrow \quad \abs{\det\!\left(\frac{\partial z_i}{\partial \mathbf u_i}\right)}
		=\sqrt{\det h_{\Sigma_i}}\abs{\cos\alpha_i}\sin\alpha_i
		=\mathcal D_i(\mathbf u_i),\quad i=1,2
	\end{align}
	where $\mathcal D_2(\mathbf u_2)$ has the same form as \eqref{eq:rho1} with $h_{\Sigma_2}$ being the induced metric on $\Sigma_2$. The flow preserves $\omega$, and the terms in $DP_{12}$ arising from variations of the travel arclength vanish upon pairing with energy-tangent vectors because $\omega(X_H,w)=-\dd H(w)=0$. Thus $P_{12}$ preserves the restricted symplectic forms on the crossing sections. The Jacobian between the two canonical coordinates $z_{1,2}$ is $\abs{\det\!\left(\frac{\partial z_2}{\partial z_1}\right)}=1$, so
	\begin{align}
		1=\abs{\det\!\left(\frac{\partial z_2}{\partial z_1}\right)}=\abs{\det\!\left(\frac{\partial z_2}{\partial \mathbf u_2}\right)}
		\abs{\det\!\left(\frac{\partial \mathbf u_2}{\partial \mathbf u_1}\right)}
		\abs{\det\!\left(\frac{\partial \mathbf u_1}{\partial z_1}\right)}
		=\mathcal D_2(\mathbf u_2)\mathcal J_{12}\mathcal D_1(\mathbf u_1)^{-1}.
	\end{align}
	This gives $\mathcal J_{12}=\mathcal D_1(\mathbf u_1)/\mathcal D_2(\mathbf u_2)$. Now the transported measure can be directly computed:
	\begin{equation}
		\dd\Gamma_{2}^{\text{trans}}=
		\left[\mathcal D_1(\mathbf u_1)\mathcal J_{12}^{-1}\right]{\dd}\mathbf u_2
		=\mathcal D_2(\mathbf u_2){\dd}\mathbf u_2,
	\end{equation}
	which has the same form as \eqref{eq:rho1} and yields \eqref{eq:intrinsic form 2} upon returning to the surface and angular variables.
}}
\begin{equation}\label{eq:intrinsic form 2}
	\dd\Gamma_2=\abs{\cos\alpha_2}\,
	\dd A_{\Sigma_2}\,\dd\Omega_{[\v{v}_2]}.
\end{equation}
This equality holds on the reached phase-space patch. Integrating over all directions at a probe point is justified only after establishing directional coverage, as we do next for the shell. 

\subsection{Application to the Schwarzschild shell}

{Now let $\M$ be the Schwarzschild shell $\mathcal A$. We set $C_1=C_+^{\mathrm{in}}$, $\pi(C_1)=\S_{+}$ and $\pi(C_2)=\Sigma$ in Theorem~\ref{thm:transport}, where $\Sigma\subset\mathcal A$ is a smooth two-sided surface and $C_2'$ is a chosen transverse intersection subset reached by geodesics launched inward from $\S_{+}$. The initial measure $\dd\Gamma_+$ in \eqref{eq:intrinsic} and~\eqref{eq:source-J} is the positive restriction of the transverse measure to the inward source section $C_+^{\mathrm{in}}$. Along each smooth intersection branch, the theorem therefore transports $\dd\Gamma_+$ to the restriction of $\dd\Gamma_\Sigma$ on the reached subset.

Recall that a shell chord may have more than one source point associated with $\S_{+}$ and may intersect $\Sigma$ more than once, with each transverse intersection counted separately. For every fixed intersection of a chord $\gamma$, there are $N_+(\gamma)$ source points representing the same unoriented chord, while the chord measure in \eqref{eq:quotient} assigns each source point the compensating weight $1/N_+(\gamma)$. After identifying opposite orientations, each source representative induces the same local crossing density. Hence the source multiplicity cancels after summing over the source points and local intersection branches, and every reachable transverse intersection is counted once. At this stage, the conclusion concerns the reachable phase-space subset over $\Sigma$, while complete directional coverage has not yet been established.}

{It remains to show that every interior point of $\mathcal A$, together with every unoriented unit direction, is reachable from $\S_{+}$. The geodesic initial-value theorem and the nontrapping argument in Section~\ref{sec:boundary} show that $\mathbb K_{\mathcal A}$ covers the shell, while Prop.~\ref{prop:outer-source-completeness} establishes that the chord projection is surjective:
\begin{equation}
	q\bigl(C_+^{\mathrm{in}}\bigr)=\mathbb K_{\mathcal A}.
\end{equation}
Combining these two statements, geodesics launched from $\S_{+}$ reach every interior point and every unoriented unit direction. Hence the normalization \eqref{eq:uniform normalization} can be applied to all directions at every interior point of $\Sigma$, giving the universal crossing flux $\pi\,\dd A_\Sigma$. We conclude that the chord measure on $\mathbb K_{\mathcal A}$ induces a crossing measure of the same intrinsic form on every smooth two-sided surface in the shell and therefore that $\mathbb K_{\mathcal A}$ uniformly covers $\mathcal A$.}

\subsection{Extension beyond the Schwarzschild case}

The construction in this section separates into a local phase-space part and a global
covering part. The local part in subsections \ref{sec:phase space} and
\ref{sec:intrinsic measure} proceeds as follows:
\begin{enumerate}
	\item The symplectic structure on $T^*\M$ gives Liouville invariance on the phase space.
	
	\item Conditioning the Liouville measure on $H=1/2$ followed by contraction with $X_H$ gives
	the oriented transverse measure $\beta$.
	
	\item Theorem \ref{thm:transport} preserves $\beta$ under the
	transporting map between the chosen open subsets $C'_1$ and $C'_2$ of transverse crossing sections.
\end{enumerate}
No Schwarzschild-specific form of the background metric is assumed, so these steps apply to general smooth Riemannian geometries. The metric remains an input and determines the Hamiltonian and geodesic flow.

On the other hand, passing from the source measure on a smooth two-sided surface $B$ to the full intersection measure on a probe surface $\Sigma$ in a target region $\mathcal{W}_{\Sigma}$ requires two global conditions for Crofton reconstruction:
\begin{enumerate}
	\item \textit{Reachability condition}: The source geodesics must cover
	all unoriented directions at every point of $\mathcal{W}_{\Sigma}$, up to sets of zero intrinsic crossing measure on each probe section.
	
	\item \textit{Multiplicity condition}: The source multiplicity $N_B(\gamma)$ must be positive and finite almost everywhere with respect to the source flux and removed by the source-multiplicity
	quotient.
\end{enumerate}
Here chords are maximal connected geodesic segments in the target region, identified under arclength translation and reversal. Writing $C_B$ for the source crossing data, let $q_B:C_B\to\mathbb K_B$ map each source datum to its corresponding maximal unoriented geodesic chord, generalizing $q$ in Eq.~\eqref{eq:map}. Here $\mathbb K_B$ denotes the resulting source-reachable chord family. With $N_B(\gamma)=\#q_B^{-1}(\gamma)$, the quotient measure is $\dd\nu_B=(q_B)_*(\dd\Gamma_B/(N_B\circ q_B))$. Each transverse unoriented source line has one oriented representative, chosen inward at a boundary source. For an interior source, all relevant intersections along the chord are included in both arclength directions. For a periodic orbit, source and probe intersections are counted over one primitive period.

Taking the Schwarzschild shell as an example, the first condition is ensured by the fact that the geodesics emanating from the outer sphere can cover the whole shell in every unoriented direction, and the second condition follows from the chord classification: each chord has one or two outer-source representatives. Nontrapping alone does not guarantee finite source multiplicity for an arbitrary source surface. Under these conditions, transport to the reachable subsets and summation over the transverse intersections give the crossing measure
\eqref{eq:intrinsic measure} and the corresponding Crofton relation.

The source multiplicity $N_B$ must be distinguished from the probe intersection count $N_\Sigma$, which counts all transverse crossing events, including repeated visits to the same spatial point. To justify the summation, assume the source projection and local intersection branches admit a countable measurable partition of transverse charts, avoiding overlaps. Since all branch contributions are nonnegative, the sum over this countable partition can be interchanged with the integral by Tonelli's theorem, even if $N_\Sigma$ is not uniformly bounded. A nonperiodic chord with infinitely many distinct source crossings is outside the finite-multiplicity construction.

These conditions allow Crofton reconstruction on a much broader class of Riemannian manifolds. A typical example is a compact smooth nontrapping Riemannian region bounded either by a closed surface $B$, or by an open surface $B$ together with a totally geodesic boundary submanifold. For this example, we require that, for almost every interior unit tangent state, the geodesic first exits the region transversely after a finite arclength when followed in either direction. In the second case, one must additionally check that almost every chord has an endpoint on $B$; total geodesicity alone does not establish this interception property. Under these conditions, the geodesics shot from $B$ cover almost every unoriented unit tangent state in the target region, with the null sets understood as specified above. The exceptional set is flow-invariant, so zero unit-energy measure $\dd\lambda$ implies zero induced flux on every transverse crossing chart. Along the intersections defined by the intrinsic measure $\dd\Gamma_B$, with the endpoint multiplicity divided out, they therefore reconstruct the areas of all smooth two-sided surfaces in the interior via the Crofton formula \cite{Santalo}.

One can state the effect of incomplete coverage without spherical symmetry. Let the source-reachability indicator $\chi_B(x,[\mathbf v])$ be one when the maximal connected geodesic chord through $(x,[\mathbf v])$ has a source representative on $B$, and zero otherwise. Whenever the transverse transport charts and finite-multiplicity quotient are available, they give
\begin{equation}
 \begin{aligned}
 \mathcal A_B[\Sigma]
 &:=\frac{1}{\pi}\int_{\mathbb K_B}
       N_\Sigma(\gamma)\,\dd\nu_B(\gamma)\\
 &=\frac{1}{\pi}\int_\Sigma\dd A_\Sigma
       \int_{\mathbb{RP}^2}\chi_B(x,[\mathbf v])
       \abs{h(\mathbf v,\mathbf n_\Sigma)}\,
       \dd\Omega_{[\mathbf v]}
 \;\leq\;\Area(\Sigma),
 \end{aligned}
 \label{eq:review-accessible-area}
\end{equation}
This is a source-dependent crossing functional, not necessarily the area of a spatial subset of $\Sigma$.

Full directional coverage up to measure-zero exceptions sets $\chi_B=1$ almost everywhere and recovers Crofton reconstruction. Missing directions of positive intrinsic crossing measure instead produce an area deficit for a transverse probe.  For a sphere $\mathcal S_r$ on the opposite side of the bridge, a unilateral source gives $\mathcal A_B[\mathcal S_r]/\Area(\mathcal S_r)=r_h^2/r^2$, as in Eq.~\eqref{eq:review-bridge-coverage}. Thus reaching every spatial point is not sufficient: the source must also supply the required directional coverage.

\section{PEE tensor network and holography beyond AdS}
\label{sec:pee-shell}

In the standard AdS/CFT correspondence, the field theory on the asymptotic boundary is equivalent to the quantum theory of gravity in the bulk AdS spacetime. An important aspect of this duality is the equivalence between the quantum entanglement structure of the boundary CFT and the spatial geometry of the AdS bulk. This equivalence is concretely realized by the partial entanglement entropy (PEE) program \cite{Wen:2018whg,Wen:2019iyq,Wen:2020ech} for vacuum AdS space: it was shown in \cite{Lin:2023rxc} that the PEE structure of the vacuum state of the boundary CFT determines a configuration of geodesics, the PEE threads, in the dual AdS space, and these threads were subsequently identified in \cite{Lin:2024dho} with the geodesics of the kinematic space. The kinematic space and its measure thus encode the area functional on the AdS spatial slice through the Crofton formula. Furthermore, Refs.~\cite{WenXuZhong,BasuWen} constructed specific tensor-network states on the PEE network, or its restriction to subregions, whose entropies reproduce the RT formula in the regimes analyzed there. The crossing measure supplies the geometric area cost, while the tensor choice determines its realization as entropy.

The preceding sections extend this geometric structure beyond vacuum AdS. On a more general Riemannian manifold $\M$, geodesics launched with the intrinsic measure of a codimension-one source surface $\Sigma$ reconstruct a subregion $\mathcal W_{\Sigma}$ when they supply its full directional coverage. These chords form an analog of the kinematic space $\mathbb{K}_{\mathcal W_{\Sigma}}$, through which the Crofton formula reconstructs the area of every surface inside $\mathcal W_{\Sigma}$. With this geometric structure at hand, the PEE tensor network previously defined in vacuum AdS space \cite{WenXuZhong,BasuWen} can be naturally extended to the reconstruction subregion $\mathcal W_{\Sigma}$ of a more generic manifold, providing a geometric basis for holographic toy models relating the boundary surface $\Sigma$ to the subregion $\mathcal W_{\Sigma}$. We also denote the geodesic chords in $\mathbb{K}_{\mathcal W_{\Sigma}}$ as the PEE threads. Although in these generalized cases we do not have the PEE structure for the quantum state on the boundary $\Sigma$, the measure for the PEE threads is determined by the intrinsic measure of the boundary surface $\Sigma$.

In this section we will explicitly construct the PEE tensor network models for the case of Schwarzschild exterior. Here $\mathcal A$ again denotes the one-sided exterior
shell of Eq.~\eqref{eq:shell}: both $\S_{-}$ and $\S_{+}$ lie in the same
Schwarzschild exterior. We first construct a pure state on
their union, then take $R_-\to r_h$ and $R_+\to\infty$ respectively. For the concrete model and its two limits, we specialize
throughout this section to the Schwarzschild metric in
Eq.~\eqref{eq:schwarzschild}, with black-hole mass $m$ and horizon radius
$r_h=2Gm$. We identify separately the steps that do not require this
specific metric. The resulting
model gives an explicit entropy calculation on a fixed background, while its
interpretation as a gravitational dual remains a hypothesis.

\subsection{The PEE tensor network restricted to the reconstruction region \texorpdfstring{$\mathcal W_{\Sigma}$}{W(Sigma)}}
\label{sec:pee-network}
For completeness, let us briefly review the definition of a PEE tensor network. A PEE tensor network is a tensor network whose bonds strictly follow the PEE threads, namely the geodesic chords in $\mathbb K_{\mathcal W_{\Sigma}}$ that perfectly tessellate the reconstructed region $\mathcal W_{\Sigma}$. In the continuum, each smooth branch of the threads emitted from a point of the boundary $\Sigma$ is described locally by a vector field whose norm gives the thread density and whose integral curves follow the threads. Multiple branches are treated separately where the spatial projection develops caustics. Combining these branches from all source points yields a network that perfectly tessellates $\mathcal W_{\Sigma}$. The measure-weighted number of threads crossing a surface is obtained by summing the absolute transverse fluxes of the source-resolved branch fields, with the source-multiplicity weights included, so that oppositely oriented crossings do not cancel. This is the continuous analogue of an ordinary tensor network.  For a finite tensor network, a cut gives an upper bound on entropy through the logarithms of its bond dimensions. Saturation by suitable tensors and a controlled relation between graph cuts and smooth surfaces are additional requirements. The Crofton identity supplies the continuum surface cost needed for such a construction; it does not alone prove entropy saturation or convergence of discrete minimal cuts.

\begin{figure}[h]
	\centering
	\includegraphics[width=0.65\linewidth]{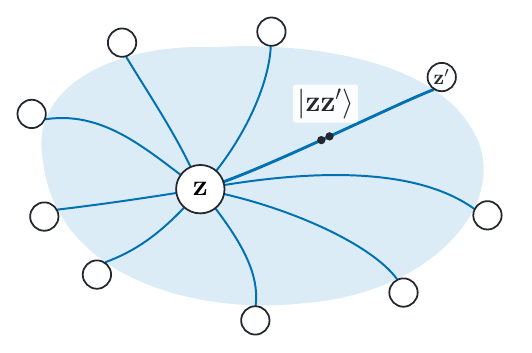}
	\caption{A tensor state $\ket{\mathcal{T}(\v z)}$ sits at the bulk site $\v z$ (hollow white dot). Its continuously distributed legs are described by a vector field (blue region), the blue lines showing representative threads. Adjacent tensors are glued by the maximally entangled state $\ket{\v z\v z'}$ along the PEE thread connecting them.}
	\label{fig:PEE network}
\end{figure}

The role of the geodesic network is to provide a geometric foundation for the state construction before a boundary theory has been identified. It expresses the area functional and the entanglement capacity of a cut using the same local crossings, making their proportionality manifest once the bond weights are fixed. Which states saturate that capacity is a quantum question, addressed separately by endpoint pairing in the special factorized examples and by minimal-cut behavior in finite random networks. The general geometric representation and the state-dependent entropy statement thus play complementary roles.

We fix the normalization before assigning a state. In three spatial dimensions, we use the geometric thread measure
\begin{equation}
 \dd n(\gamma):=\frac{\dd\nu(\gamma)}{\pi},
 \qquad
 \mathcal N_{\mathrm{surf}}(\Gamma)
 :=\int N_\Gamma(\gamma)\,\dd n(\gamma)
 =\Area(\Gamma).
 \label{eq:review-thread-normalization}
\end{equation}
Thus $\mathcal N$ below is an area-normalized geometric count, not a raw number of qudits. For the two-dimensional illustrations below, the same convention uses $\dd n=\dd\nu/c_2$ with $c_2=2$, so that the crossing count reproduces length rather than area. For a discrete bundle of measure $\delta n_e$, the entropy capacity assigned to a bond is $\log D_e\simeq\delta n_e/(4G)$ in units with $\hbar=1$. This convention keeps the Crofton factor $1/\pi$ distinct from the gravitational normalization $1/(4G)$. The latter is an input to the toy model. The continuum products below describe its formal geometric limit; a Hilbert-space state and its entropy are defined using finite regulated networks.

Let $B$ denote the complete boundary carrying uncontracted legs. When there are no additional open boundary components, $B=\Sigma$ and we may write $\ket{\Psi_\Sigma}$. If $B=\Sigma\cup\gamma_\Sigma$, the full network state $\ket{\Psi_B}$ is pure, whereas the state on $\Sigma$ is obtained by tracing over $\gamma_\Sigma$. This distinction will be important when counting entanglement with horizon legs.

At each bulk site $\v z$ one places a local tensor
\begin{equation}\label{eq:pee-vertex}
	\ket{\mathcal{T}(\v z)}=\mathcal{T}_{u_1u_2\cdots}(\v z)\ket{u_1}_{\v z}\ket{u_2}_{\v z}\cdots,
\end{equation}
where the index $u_i$ labels the legs carried by the threads through $\v z$ (in the continuum $i$ is a continuous parameter).  See Fig.~\ref{fig:PEE network} for illustration. Any two adjacent sites $\v z,\v z'$ are glued by projecting the corresponding pair of legs onto a maximally entangled state $\ket{\v z\v z'}$, so that the network represents a state on the uncontracted legs of $B$,
\begin{equation}\label{eq:pee-network-state}
	\ket{\Psi_B}=\left(\bigotimes_{\langle\v z,\v z'\rangle}\bra{\v z\v z'}\right)\left(\bigotimes_{\v z\in \mathcal W_{\Sigma}}\ket{\mathcal{T}(\v z)}\right).
\end{equation}
Here $\langle\v z,\v z'\rangle$ runs over pairs of adjacent sites connected by a thread bond.

Depending on the tensors placed at the sites, one obtains different types of tensor network models. We consider the factorized model and the random PEE tensor network. The simplest is the factorized model, in which each site carries a tensor product of EPR pairs, one pair for each thread passing through it,
\begin{equation}\label{eq:pee-factorized-vertex}
	\ket{\mathcal{T}(\v z)}=\bigotimes_{\{i,j\}}\mathcal{T}_{ab}(\v z)\ket{a}_{\v z i}\ket{b}_{\v z j},
\end{equation}
where $\{i,j\}$ labels a pair of legs on the same thread and $\mathcal{T}_{ab}(\v z)$ is an EPR pair, $\mathcal{T}^\dagger\mathcal{T}\propto\mathbb{I}$. Contracting the pairs along a thread iteratively leaves a single EPR pair on that thread, so the full network reduces to a state $\ket{\Psi_B}$ given by a tensor product of EPR pairs on $B$, one for each PEE thread $\gamma\in\mathbb{K}_{\mathcal W_{\Sigma}}$. The two qudits of each pair reside at the two endpoints of the corresponding thread. Explicitly,
\begin{equation}\label{eq:pee-boundary-state}
	\ket{\Psi_B}=\bigotimes_{\gamma\in \mathbb{K}_{\mathcal W_{\Sigma}}}\ket{\mathrm{EPR}}_{\gamma}=\bigotimes_{\gamma\in \mathbb{K}_{\mathcal W_{\Sigma}}}T_{ab}(\gamma)\ket{a}\ket{b},
\end{equation}
where $T_{ab}(\gamma)\ket{a}\ket{b}$ represents the EPR pair on the two endpoints of $\gamma$. 

\begin{figure}[h]
	\centering
	\includegraphics[width=1.0\linewidth]{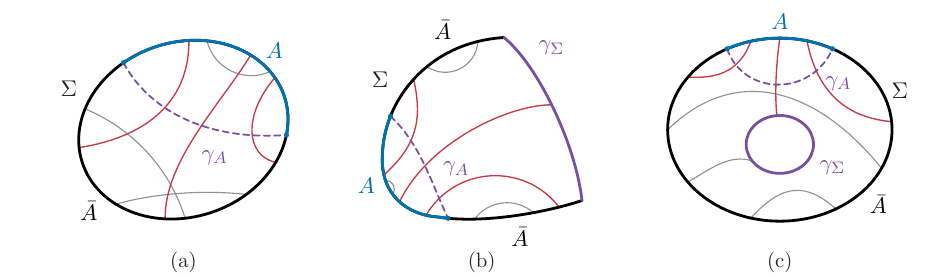}
	\caption{We show three examples of two-dimensional $\mathcal{W}_{\Sigma}$, which is the region bounded by $\Sigma$ and $\gamma_\Sigma$ (the purple solid line). We split $\Sigma$ into $A$ (the blue solid line) and $ \bar A$ (the black solid line).  In the illustrated configurations, a thread has endpoints in $A$ and its full complement $\bar A\cup\gamma_\Sigma$ (red lines) precisely when it crosses the separator $\gamma_A$ once; the normalized crossing measure then equals the length of $\gamma_A$. This endpoint--crossing correspondence is an additional geometric property of these examples. The gray PEE threads represent the EPR pairs confined in $A$ or $\bar{A}$, which do not contribute to $S_A$.}
	\label{fig:RT-illustration}
\end{figure}

We now illustrate the endpoint-pair entropy in the two-dimensional examples of Fig.~\ref{fig:RT-illustration}. The complete boundary is $B=\Sigma\cup\gamma_\Sigma$, with $\gamma_\Sigma=\emptyset$ when there is no supplementary component. For $A\subset\Sigma$, the complement relevant to entropy is therefore $B\setminus A=\bar A\cup\gamma_\Sigma$. In the regulated factorized state, only EPR pairs whose endpoints lie on opposite sides of this partition contribute. With the normalization of Eq.~\eqref{eq:review-thread-normalization}, write their geometric measure as $\mathcal N_{\mathrm{end}}(A,B\setminus A)$. The endpoint count agrees with the crossing count of a minimum-length homologous separator $\gamma_A$ if every connecting chord crosses it exactly once and no other chord crosses it. These are the sufficient conditions used in the illustrated examples. Total geodesicity alone should not be taken to guarantee this global intersection property on an arbitrary Riemannian manifold. When it holds, $S_{\mathrm{fact}}(A)=\mathcal N_{\mathrm{end}}(A,B\setminus A)/(4G)$ gives
\begin{equation}
	S(A)=\frac{\Area(\gamma_A)}{4G}\,,
	\label{eq:pee-RT}
\end{equation}
which is precisely the RT formula.

The once-crossing property is restrictive. More generally, any separating surface must meet every chord whose endpoints belong to opposite boundary subsystems, but may also meet other chords or count a chord repeatedly. Consequently the factorized model obeys
\begin{equation}
 S_{\mathrm{fact}}(A)
 =\frac{\mathcal N_{\mathrm{end}}(A,B\setminus A)}{4G}
 \leq\frac{1}{4G}\inf_{\Gamma\sim A}\Area(\Gamma),
 \label{eq:review-factorized-bound}
\end{equation}
where the infimum is over admissible homologous separators. Equality needs the endpoint--crossing correspondence above. This explains why a uniform Crofton measure alone does not make a product of endpoint pairs reproduce general holographic entropies. For multi-interval configurations and in higher-dimensional manifolds, the RT surface of a general subsystem $A$ need not be totally geodesic, and endpoint pairing generally does not reproduce its area. To obtain minimal-cut behavior beyond these special factorized examples, we consider the \emph{random PEE tensor network} of Ref.~\cite{WenXuZhong}, in which each bulk site carries a random state while adjacent sites are still glued by maximal entanglement along the PEE threads. The essential features of this construction are the following.
\begin{itemize}
	\item \textit{Structure independence.} A random tensor network can be defined on an arbitrary finite graph without requiring background homogeneity. Its adjacency and bond dimensions determine the available cuts and their costs. The construction can therefore be applied to a finite discretization of the boundary-sourced PEE thread network of $\mathcal W_{\Sigma}$.

	\item \textit{Minimal-cut entropy.} On a fixed finite graph, random tensors give a leading entropy controlled by the graph minimal cut in the appropriate large-bond-dimension regime \cite{HaydenEtAl}. Its cost is the sum of $\log D_e$ over cut bonds. This is a statement about graph partitions, distinct from direct endpoint pairing and from minimization over smooth bulk surfaces.

	\item \textit{Continuum identification requires control.} Crofton reconstruction fixes the cost of every fixed smooth separator. To identify the graph minimal cut with the continuum minimal area, one must also exclude graph partitions with anomalously small cost that have no appropriate smooth limit. Pointwise convergence of crossing counts for fixed surfaces is insufficient for this minimization problem.
\end{itemize}

The random construction therefore supplies a candidate route to the RT formula for general regions, subject to a regulator with the required convergence properties. A sufficient program would approximate the weighted chord measure by finite networks so that each fixed smooth separator has the Crofton cost, prove a lower bound for all competing graph cuts in terms of limiting geometric separators, and construct graph cuts approaching the continuum minimum. Compactness or another control of minimizing sequences is also needed when taking the continuum limit. These lower-bound and approximation requirements explain the step not supplied by the fixed-surface Crofton identity. They are conditions for the RT interpretation here, rather than a convergence theorem established in this paper.

The large-bond-dimension limit must likewise be distinguished from refinement of the graph. Refinement reduces the measure represented by an individual bundle, whereas the random-tensor limit requires its bond dimension to be large. A joint limit in which $G$ decreases sufficiently rapidly relative to the bundle weights is one possible prescription, but needs quantitative estimates. Until a regulator and error bounds are specified, the general-region RT relation is a conditional leading-order statement. The endpoint-pair entropies for the horizon and the two complete asymptotic boundaries, evaluated next, do not require this minimization step.

Following Ref.~\cite{BasuWen}, the construction suggests a surface/state interpretation \cite{Miyaji:2015yva,Miyaji:2015fia} and a model for generalized entanglement wedges \cite{Bousso:2022hlz,Bousso:2023sya} on the reconstruction region. Establishing an actual encoding of bulk operators, the recoverable algebra for a boundary subsystem, and its errors requires additional tensor-network input. A geometric area identity and a boundary tensor state alone do not establish a bulk gravitational duality.

\subsection{Holographic toy model for Schwarzschild background}
Now we apply the PEE tensor network model of the previous subsection to the Schwarzschild black hole and the Einstein--Rosen bridge.

For the Schwarzschild black hole, we take the outer sphere $\S_{+}$ to infinity and push the inner sphere $\S_{-}$ to the horizon, so that the entire exterior region becomes the reconstruction region of $\S_{+}$. The PEE tensor network then covers this region, and $\S_{+}$ together with the horizon $\mathcal S_H$ supports a pure state. In the factorized model, the entanglement entropy of $\S_{+}$ is obtained by counting the PEE threads connecting it to the horizon. Each such thread has one horizon endpoint, and their normalized measure reproduces the horizon area. With the entropy weight $1/(4G)$, this gives
\begin{equation}
	S(\S_{+})=\frac{1}{4G}\,\mathcal N(\S_{+},\mathcal S_H)=\frac{A_H}{4G},
	\label{eq:pee-bh-boundary}
\end{equation}

Here $\mathcal N=\mathcal N_{\mathrm{end}}$ uses the normalized measure $\dd n$ of Eq.~\eqref{eq:review-thread-normalization}, and $A_H=4\pi r_h^2$. The horizon endpoint count is exact in the regulated factorized construction and tends to this geometric value in its measure limit. It reproduces the Bekenstein--Hawking value \cite{Bekenstein,Hawking:1975vcx} after the entropy normalization has been chosen; this is not an independent derivation of the coefficient $1/(4G)$ or of microscopic black-hole degeneracy. For a subregion $A\subset\S_+$, let $\bar A=\S_+\setminus A$. Since the complete pure-state boundary is $B=\S_+\cup\mathcal S_H$, the factorized complement includes the horizon legs. The two entropy prescriptions must therefore be written separately:
\begin{equation}
	\begin{aligned}
 &S_{\mathrm{fact}}(A)=\frac{\mathcal N_{\mathrm{end}}(A,\bar A\cup\mathcal S_H)}{4G},\\
 &S_{\mathrm{rand}}(A)\sim\frac{\Area(\gamma_A)}{4G}
 \quad\text{under the continuum and large-bond assumptions}.
\end{aligned}
	\label{eq:pee-subregion}
\end{equation}
Here $\gamma_A$ is a minimum-area separator obeying the homology condition in the regulated shell, with the other open components assigned to the complement. The factorized expression equals this area only under the endpoint--crossing condition described above.

For the Einstein--Rosen bridge, the two boundary sources tessellate the whole bridge, and we construct a pure state on their union. The threads connecting the two boundaries are the crossing geodesics with $0<J<r_h$. Launched from one side, each passes through the throat and emerges on the other side, intersecting the throat exactly once. Their total measure, Eq.~\eqref{eq:cross-nu}, divided by the Crofton factor $\pi$ as in Eq.~\eqref{eq:review-thread-normalization}, reproduces the throat area $A_H=4\pi r_h^{2}$. Assigning each crossing thread the entanglement weight $1/(4G)$ and using the factorization above, the entanglement entropy between the two boundary systems is
\begin{equation}
	S_{\mathrm{LR}}=\frac{A_H}{4G},
	\label{eq:pee-LR}
\end{equation}
i.e. the Bekenstein--Hawking entropy of the black hole \cite{Bekenstein,Hawking:1975vcx}. This is a concrete realization of ER$=$EPR \cite{MaldacenaSusskind,JensenKarch,Jiang:2024er} in the factorized model: the threads joining the two boundaries are the EPR pairs, while the throat they traverse is the Einstein--Rosen bridge.\footnote{Here ER$=$EPR is understood as a reinterpretation of the Einstein--Rosen bridge geometry in terms of EPR pairs within the present model; it does not address the black-hole information paradox \cite{Hawking:1976information,Page:1993information}. Extending the geodesic-network construction and the associated PEE tensor-network models to geometries in the island phase \cite{Penington:2019reconstruction,Almheiri:2019evaporating,1908.10996} is an important direction for future research.} The shared crossing sector thus carries precisely the left--right entanglement capacity assigned to the throat. The private returning sectors affect correlations within each side but do not contribute to the entropy between the two complete boundary systems. This realizes a relation between a supplied bridge geometry and EPR connectivity within the model, without deriving the bridge metric from entanglement dynamics.

The one-sided and two-sided entropy calculations use the same geometric resource. Cutting each crossing chord at the throat gives a horizon leg that purifies the outer system; joining the corresponding left and right pieces instead gives an EPR pair between the asymptotic systems. The normalized measure of this sector is $A_H$, counted once, so both bipartitions have entropy $A_H/(4G)$. In the factorized model, Bekenstein--Hawking scaling and the ER=EPR picture are two uses of this common crossing measure, rather than independent microscopic explanations of black-hole thermodynamics. The additional content supplied by the geometry is that this sector is sufficient to saturate the throat area while remaining insufficient for complete directional coverage in the opposite exterior.

Finally, we stress that the proposed holographic interpretation is restricted to the PEE tensor-network model constructed here, and it is not clear what the theory on the sphere $\S_{+}$ at infinity actually is. For more general backgrounds, the source boundary need not be spherical. At finite $R_+$, a timelike worldtube $\S_+\times\mathbb R$ can organize the regulated boundary degrees of freedom. The present spatial construction specifies neither a field theory on that worldtube nor its time evolution, and the limit to spatial infinity is not the timelike conformal boundary of AdS. It should be distinguished from null-boundary proposals for flat-space holography and holographic entanglement entropy \cite{Bagchi:2025vri,Nguyen:2025zhg,Ruzziconi:2026bix,Bagchi:2010zz,Bagchi:2012cy,Bagchi:2012xr,Barnich:2012aw,Barnich:2012rz,Bagchi:2014iea,Jiang:2017ecm,Apolo:2020bld,Wen:2018mev,Laddha:2020hni}. In the existing literature, this places our construction closer to the spatial-infinity holographic picture (see, e.g., Refs.~\cite{LiTakayanagiFlat,GhoshKrishnanSpi,Krishnan:2019causal,Henneaux:2018spatial,Henneaux:2019review}) than to the null-boundary proposals. Connections with these approaches, as well as with flat-space entanglement in Coulomb-branch geometries \cite{Jorstad:2026flat}, remain to be explored.

\section{Discussion}
\label{sec:conclusion}

This work develops a geodesic-network representation of geometry and its associated entanglement models in three steps. First, using the Schwarzschild exterior and Einstein--Rosen bridge as concrete examples, we construct boundary-sourced geodesic networks that perfectly tessellate bulk subregions, extending beyond the homogeneous manifolds considered in earlier PEE tensor-network models. Second, we formulate this construction on general smooth Riemannian backgrounds and characterize the subregions reconstructible from a chosen source surface through directional coverage and source multiplicity. Local Liouville transport preserves the crossing measure, while these global conditions determine where perfect tessellation holds. Third, building on these purely geometric results, we extend the PEE tensor-network models previously defined on homogeneous backgrounds to the resulting networks. The Schwarzschild models reproduce the Bekenstein--Hawking area law and an ER$=$EPR description of the bridge. They also reproduce the RT formula in the factorized configurations identified in Section~\ref{sec:pee-network}; its extension to general regions uses random tensors and remains subject to the graph-to-continuum conditions stated there.

We next clarify how the geometric part of our construction relates to the existing mathematics literature. The Crofton formula on homogeneous spaces, where the kinematic measure is inherited from the isometry group, is classical \cite{Santalo,Helgason}. This classical result underlies the kinematic-space formulation of AdS/CFT \cite{CzechEtAl,Czech:2015tensor,Czech:2016stereo,Czech:2017zfq,Czech:2019time} and the PEE-thread program \cite{Wen:2018whg,Wen:2019iyq,Wen:2020ech,Lin:2023rxc,Lin:2024dho,WenXuZhong,BasuWen}. The local ingredients of our transport argument (the Liouville measure on the unit tangent bundle, its invariance under the geodesic flow, and the transverse measure obtained by removing the flow direction) are likewise classical and independent of the specific form of the metric. In the mathematics literature, a general Crofton-type formalism for non-homogeneous Riemannian and Finsler manifolds, built from the symplectic structure on the space of geodesics, was developed by \'Alvarez Paiva and collaborators \cite{AlvarezPaivaFernandes,AlvarezPaivaBerck,AlvarezPaivaHilbert}. In particular, a general Crofton formula for Finsler spaces via the Holmes--Thompson volume was established in \cite{AlvarezPaivaBerck}. Within this classical framework, our geometric contribution is to resolve the global coverage and multiplicity requirements of the boundary-source construction for the Schwarzschild exterior and bridge: which sectors a chosen source supplies, why a unilateral description loses the opposite returning sector, and how bilateral sources restore a single correctly normalized chord measure. Their entanglement interpretation relies on this resolved sector structure, rather than on a new local invariance theorem.

Our construction shows that a geometrically faithful tensor-network description need not be restricted to homogeneous gravitational backgrounds. At the level of the continuum crossing measure, the geodesic network exactly reproduces the area of every smooth hypersurface wherever the coverage conditions hold. The geometric content is fixed directly by the Crofton measure, so the background area functional is built into the network. Placing PEE tensors on this network therefore extends the geometric foundation of holographic tensor-network models to a broad class of inhomogeneous gravitational spaces. This provides a quantitative framework for studying geometry and entanglement on supplied inhomogeneous gravitational backgrounds, rather than tying the applicability of tensor networks to the symmetries of vacuum AdS.\footnote{A complementary development is Ref.~\cite{Wen:2026replica}, which constructs replica-dependent geodesic tensor networks that reproduce the cosmic-brane prescription for refined R\'enyi entropy while retaining maximally entangled links. The replica-index dependence is carried by the changing geometry rather than by modifying the link states.}

The same framework also offers an entanglement-based representation of spatial geometry: geometric area and network entanglement capacity are expressed by the same local crossing measure. With suitable quantum states and entropy weights, this common crossing description makes the entropy--area correspondence manifest within the network. This perspective provides a concrete paradigm for organizing the entanglement resources associated with a gravitational geometry and a starting point for investigating how such geometries can be constructed from quantum-entanglement data.

Our results also indicate that bulk matter need not act as a source of geodesic threads: a sourceless geodesic gas can tessellate a Riemannian geometry supported by a matter distribution, since Liouville conservation holds for geodesic motion on that metric. Opening geodesic legs in the bulk instead introduces additional bulk degrees of freedom, whose entanglement with the boundary must be specified as part of the network state. Two directions remain open for future work. The first is to lift the construction from time-independent spatial slices to genuinely time-dependent spacetimes, which requires a Lorentzian generalization with an appropriate family of geodesics and an invariant crossing measure. The second is to go beyond pure geometry and include quantum excitations, promoting the classical tessellation of a fixed background to a quantum description. Following Ref.~\cite{WenXuZhong}, the quantum state at a bulk point is a tensor built from the quantum states on the open legs of the geodesics passing through that point, and neighboring points are glued by maximal entanglement. The network state with all bulk legs contracted could provide a model for a ground state of quantum gravity. Releasing this gluing at selected locations opens new bulk legs that can carry candidate quantum excitations; their dynamics and backreaction require an additional dynamical principle. 

\appendix

\section{Proof of the transport theorem}
\label{sec:app-transport}

In this appendix we state and prove the transport theorem used in the main text. Here $n=\dim\M$ denotes the dimension of the underlying Riemannian manifold. We write the proof for $n=3$; in general dimension $n$, the four test vectors below are replaced by $2n-2$ vectors.

\begin{thm}\label{thm:transport}
	Consider two transverse phase-space crossing sections $C_{1,2}\subset\Sbb_{H=1/2}\subset T^*\M$, of codimension one in this regular energy surface, with $C_{1,2}'\subset C_{1,2}$ open subsets on which the travel arclength $\tau(\xi)$ required for the geodesic-flow orbit starting at $\xi\in C'_1$ to reach a chosen intersection with $C'_2$ is smoothly well-defined. Assume that the flow is defined along every such connecting segment. Then the oriented transverse measure $\beta$ is invariant under the transporting map $P_{12}(\xi)=\varphi_H^{\tau(\xi)}(\xi)$, $\xi\in C'_1$. That is, Eq.~\eqref{eq:transporting map 1} holds.
\end{thm}

\begin{proof}
	For vector fields $X,Y$ and any differential form $\eta$, the Lie derivative and contraction obey
	\begin{equation}
		\mathscr{L}_X\iota_Y\eta
		-\iota_Y\mathscr{L}_X\eta
		=\iota_{[X,Y]}\eta,
		\label{eq:Lie contraction}
	\end{equation}
	where $[X,Y]$ is the Lie bracket.  The oriented transverse measure $\beta$ is therefore invariant under the fixed-arclength Hamiltonian flow:
	\begin{equation}
		\mathscr{L}_{X_H}\beta
		=\mathscr{L}_{X_H}\iota_{X_H}\dd\lambda
		=\iota_{X_H}\mathscr{L}_{X_H}\dd\lambda
		+\iota_{[X_H,X_H]}\dd\lambda
		=0,
	\end{equation}
	where the two terms vanish separately because of \eqref{eq:Liouville invariance dlambda} and $[X_H,X_H]=0$, respectively.  By the definition of the Lie derivative in \eqref{eq:Lie derivative from Hamiltonian flow}, this statement is equivalent to
	\begin{equation}\label{eq:transported measure 1}
		(\varphi_H^s)^*\beta=\beta
	\end{equation}
	for all fixed $s$ for which the local flow is defined.
	
	The transporting map \eqref{eq:transporting map}, however, evolves different initial points through different arclengths $\tau(\xi)$.  We now show directly that this variation does not change the transverse measure.  Introduce the two-variable flow map $\Psi(s,\xi)$, where $\xi=(x,p)\in\Sbb_{H=1/2}$ denotes a full phase-space state, i.e. a base point $x\in\M$ together with a unit-speed momentum covector $p$ that encodes the direction, by
	\begin{equation}
		\Psi:\Rbb\times\Sbb_{H=1/2}\longrightarrow\Sbb_{H=1/2},
		\qquad
		\Psi(s,\xi)=\varphi_H^s(\xi),
	\end{equation}
	on the neighborhood where the local flow is defined, and let $g(\xi)=(\tau(\xi),\xi)$.  Then $P_{12}=\Psi\circ g$.    For $w\in T_\xi(C_1')$ \footnote{Here $T_\xi(C_1')$ is the tangent space of the crossing section $C_1'$ at $\xi$, i.e. the space of infinitesimal perturbations of the state $\xi$ that stay within $C_1'$ (shifting the base point along $\Sigma_1$ and tilting the unit direction).}, the differential of the map $g$ reads
	\begin{equation}
		Dg|_\xi(w)=\bigl(\dd\tau|_\xi(w),w\bigr),
	\end{equation}
	where the first component $\dd\tau|_\xi(w)\in\mathbb R$ is the rate at which the travel arclength $\tau$ changes along the perturbation $w$, while the second component $w$ is the variation of the state $\xi$ itself, i.e. the identity part of $g$. The differential of $\Psi$ is
	\begin{equation}\label{eq:differential of Psi}
		D\Psi|_{(s,\xi)}(a,w)
		=a\,\partial_s\Psi(s,\xi)+D_\xi\Psi(s,\xi)(w).
	\end{equation}
	Here $D_\xi$ means that the arclength parameter $s$ is held fixed while the initial point $\xi$ is varied.  Thus $D_\xi\Psi(s,\xi)$ maps an initial perturbation $w\in T_\xi\Sbb_{H=1/2}$ to the resulting perturbation at $\Psi(s,\xi)$ after the same arclength $s$.  It is distinct from $\partial_s\Psi(s,\xi)$, which varies the arclength parameter while keeping the initial point $\xi$ fixed.  The factor $a$ appears because it is the actual rate of change in the $s$ direction.  Indeed, choose a curve $(s(\epsilon),\xi(\epsilon))$ with $\dot s(0)=a$ and $\dot\xi(0)=w$.  The ordinary chain rule gives
	\begin{equation}
		\left.\frac{\dd}{\dd\epsilon}
		\Psi\bigl(s(\epsilon),\xi(\epsilon)\bigr)
		\right|_{\epsilon=0}
		=a\,\partial_s\Psi(s,\xi)
		+D_\xi\Psi(s,\xi)(w).
	\end{equation}
	Thus $\partial_s\Psi$ is the change per unit increase of $s$, while multiplication by $a={\dd}s/\dd\epsilon$ gives the change along the chosen tangent vector.
	
	Since $\partial_s\varphi_H^s(\xi)=X_H|_{\varphi_H^s(\xi)}$, the chain rule $DP_{12}|_\xi=D\Psi|_{g(\xi)}\circ Dg|_\xi$ applied to $P_{12}=\Psi\circ g$ gives
	\begin{equation}\label{eq:differential of P12}
		DP_{12}|_\xi(w)
		=\dd\tau|_\xi(w)\,X_H|_{P_{12}(\xi)}+D_\xi\varphi_H^{\tau(\xi)}(w),
	\end{equation}
	where the first term records the variation of the travel arclength, while the second records the variation of the initial point at fixed travel arclength.
	
	We now relate \eqref{eq:differential of P12} to the pullback of the oriented transverse measure $\beta$. Choose four linearly independent tangent vectors $w_1,\ldots,w_4\in T_\xi(C_1')$, which span the $4$-dimensional tangent space $T_\xi(C_1')$ and thereby define an infinitesimal $4$-cell in the crossing section $C_1'$. The definition of pullback gives
	\begin{equation}\label{eq:pullback1}
		\begin{aligned}
			&\left[P_{12}^*\left(\left.\beta\right|_{C_2'}\right)\right]_\xi
			(w_1,\ldots,w_4)=\beta_{P_{12}(\xi)}
			\bigl(DP_{12}|_\xi(w_1),\ldots,
			DP_{12}|_\xi(w_4)\bigr).
		\end{aligned}
	\end{equation}
	Here $\beta|_{C_2'}$ is the restriction of the $4$-form $\beta$ to the crossing section $C_2'$, which acts only on tangent vectors of $C_2'$, whereas $\beta_{P_{12}(\xi)}$ is the value of $\beta$ at the single point $P_{12}(\xi)\in C_2'$, acting on the whole tangent space $T_{P_{12}(\xi)}\Sbb_{H=1/2}$. On the arguments $DP_{12}|_\xi(w_j)$, which are tangent to $C_2'$, the two coincide. Thus \eqref{eq:pullback1} is just the definition of the pullback: it pushes the four tangent vectors $w_j$ forward to the target point $P_{12}(\xi)$ and applies $\beta$ there, evaluating the oriented transverse measure of the infinitesimal $4$-cell transported from $\xi$ to $P_{12}(\xi)$. Equivalently, since $\beta$ is the flux density of the geodesic gas, \eqref{eq:pullback1} computes the infinitesimal flux through the transported $4$-cell. Concretely, the left-hand side of \eqref{eq:pullback1} represents this target flux pulled back to $C_1'$, while the right-hand side evaluates it directly on $C_2'$. Once the theorem is established, the equality of this pullback with $\beta|_{C_1'}$ becomes the statement that the geodesic flux is conserved along the geodesic flow (equivalently, that the geodesic gas has no sources or sinks in the interior), in agreement with the setup of Section~\ref{sec:2}.
	
	Set
	\begin{equation}
		A_j\equiv D_\xi\varphi_H^{\tau(\xi)}(w_j),
		\qquad
		a_j\equiv \dd\tau|_\xi(w_j).
	\end{equation}
	Equation \eqref{eq:differential of P12} reads $DP_{12}|_\xi(w_j)=A_j+a_jX_H$, where we use $X_H$ to denote $X_H|_{P_{12}(\xi)}$ for simplicity.  Since
	\begin{equation}
		\iota_{X_H}\beta
		=\iota_{X_H}\iota_{X_H}\dd\lambda=0,
	\end{equation}
	the contraction identity means
	\begin{equation}
		(\iota_{X_H}\beta)(Y_1,Y_2,Y_3)
		=\beta(X_H,Y_1,Y_2,Y_3)=0
	\end{equation}
	for arbitrary tangent vectors $Y_1,Y_2,Y_3$.  If $X_H$ occurs in another slot, the antisymmetry of $\beta$ moves it to the first slot.  By the multilinearity and antisymmetry of $\beta$, we therefore have
	\begin{equation}\label{eq:pullback2}
		\begin{aligned}
			\beta_{P_{12}(\xi)}
			\bigl(A_1+a_1X_H,\ldots,A_4+a_4X_H\bigr)&=
			\beta_{P_{12}(\xi)}(A_1,\ldots,A_4)\\&=
			\left[\left(\varphi_H^{\tau(\xi)}\right)^*\beta\right]_\xi
			(w_1,\ldots,w_4)\\
			&=\beta|_\xi
			(w_1,\ldots,w_4),
		\end{aligned}
	\end{equation}
	where in the first line every term containing at least one $X_H$ vanishes, in the second line we use the definition of pullback again, and in the third line we use \eqref{eq:transported measure 1}. Note that the pullback in the second line is by the fixed-arclength flow $\varphi_H^{\tau(\xi)}$, not by the transporting map $P_{12}$. Once the $X_H$-terms are dropped in the first line, only the fixed-arclength part $D_\xi\varphi_H^{\tau(\xi)}$ survives, which is precisely the differential of $\varphi_H^{\tau(\xi)}$.  Finally, combining \eqref{eq:pullback1} and \eqref{eq:pullback2} yields
	\begin{equation}
		P_{12}^*\left(\left.\beta\right|_{C_2'}\right)
		=\left.\beta\right|_{C_1'}.
		\label{eq:Poincare transport}
	\end{equation}
	Thus the oriented transverse measure $\beta$ is invariant under the transporting map from $C_1'$ to $C_2'$, even though the travel arclength varies across the section. Taking absolute densities gives the corresponding invariance of the positive crossing measure.
	
\end{proof}


\section*{Acknowledgements}

The authors are supported by the National Natural Science Foundation of China (NSFC) Grant No. 12447108 and the Shing-Tung Yau Center of Southeast University. The authors thank D. Basu and M. Xu for related collaboration.
\bibliography{schwarzschild_crofton}

\end{document}